\documentclass[11pt,letterpaper]{article}
\usepackage{epsfig,rotating,setspace,latexsym,amsbsy,amsmath,epsf,amssymb,bm}
\usepackage{cite,graphicx,authblk,color,subfigure}
\usepackage[title]{appendix}
\usepackage{multirow}
\usepackage{ctable}
\usepackage{nicefrac}

\newlength{\Oldarrayrulewidth}

\usepackage{amsmath}
\usepackage{bm}
\usepackage{enumerate}

\newcommand{\W}{{\mathbf{W}}}

\newtheorem{theorem}{Theorem}
\newtheorem{lemma}{Lemma}

\newtheorem{remark}{Remark}

\newtheorem{proposition}[theorem]{Proposition}
\newtheorem{corollary}{Corollary}

\newenvironment{proof}[1]{\medskip\par\noindent
{\bf Proof:\,}\,#1}{{\mbox{\,$\blacksquare$}\par}}
\allowdisplaybreaks

\title{Asymmetric Leaky  Private Information Retrieval}
\author{Islam Samy \quad  Mohamed A. Attia \quad Ravi Tandon \quad Loukas Lazos\\
Department of Electrical and Computer Engineering\\
University of Arizona, Tucson, AZ, USA. \\
Email: \{\textit{islamsamy, madel, tandonr, llazos}\}@email.arizona.edu\footnote{The work of I. Samy and L. Lazos was supported by NSF grant CNS 1813401. The work of M. Attia and R. Tandon was supported by NSF Grants CAREER 1651492, CNS 1715947, and the 2018 Keysight Early Career Professor Award. This work was presented in part at the 2019 IEEE International Symposium on Information Theory.}}

\begin{document}

\maketitle
\newcommand\blfootnote[1]{%
  \begingroup
  \renewcommand\thefootnote{}\footnote{#1}%
  \addtocounter{footnote}{-1}%
  \endgroup
}

%\affiliation{
	%	Department of Electrical and Computer Engineering, The University of Arizona}
%\email{Email: \{islamsamy@email, tandonr@email, and llazos@ece, \}.arizona.edu}

\maketitle
\vspace{-15pt}
\begin{abstract}
Information-theoretic formulations of the private information retrieval (PIR) problem have been investigated  under a variety of scenarios. Symmetric private information retrieval (SPIR) is a variant where a user is able to privately retrieve one out of $K$ messages  from $N$ non-colluding replicated databases without learning anything about the remaining $K-1$ messages. However, the goal of perfect privacy can be too taxing for certain applications. In this paper, we investigate if the information-theoretic capacity of SPIR (equivalently, the inverse of the minimum download cost) can be increased by relaxing both user and DB privacy definitions. Such relaxation is relevant in applications where  privacy can be traded for communication efficiency.

We introduce and investigate the Asymmetric Leaky PIR (AL-PIR) model with different privacy leakage budgets in each direction.
For user privacy leakage, we bound the probability ratios between all possible realizations of DB queries by a function of a non-negative constant $\epsilon$. For DB privacy, we bound the mutual information between the undesired messages, the queries, and the answers, by a function of a non-negative constant $\delta$.
We propose a general AL-PIR scheme that achieves an upper bound on the optimal download cost for arbitrary  $\epsilon$ and $\delta$. We show that the optimal download cost of AL-PIR is upper-bounded as $D^{*}(\epsilon,\delta)\leq 1+\frac{1}{N-1}-\frac{\delta e^{\epsilon}}{N^{K-1}-1}$. 
  Second, we obtain an information-theoretic lower bound on the download cost as $D^{*}(\epsilon,\delta)\geq 1+\frac{1}{Ne^{\epsilon}-1}-\frac{\delta}{(Ne^{\epsilon})^{K-1}-1}$.  The gap analysis between the two bounds shows that our AL-PIR scheme is optimal when $\epsilon =0$, i.e., under perfect user privacy and it is optimal within a maximum multiplicative gap of $\frac{N-e^{-\epsilon}}{N-1}$ for any $\epsilon>0$ and $\delta>0$.
%For $\epsilon>0$, we explore the opportunities offered for reducing the download cost.  
%Our bounds match for extreme values of $\epsilon$, and are within a constant multiplicative factor of $2$, for any $\epsilon$.  

\end{abstract}
\vspace{-18pt}
\section{Introduction}
In the era of big data and data analytics, users who access a plethora of online services face serious privacy risks. Their online behavior and data access patterns can be analyzed to reveal sensitive personal information and breach their privacy \cite{chor1995private}. One possible solution to such data leakages is to retrieve information privately by executing a private information retrieval (PIR) protocol. In a PIR protocol, the identity of the message retrieved by the user remains secret from the database(s). This is typically achieved at the expense of an increased communication cost to ensure that the desired message remains hidden among others.  
In the pioneering work by Chor {\em et al.} \cite{chor1995private}, the authors considered one-bit long messages. The overhead was calculated as the sum of the queries sent by the user (upload cost) and the answers provided by the database (download cost). Under arbitrarily large messages, the download cost becomes the dominant factor of the PIR overhead. This allows the PIR rate to be defined as the ratio of the message size to the number of downloaded bits. The maximum of these rates is referred to as the PIR capacity and its reciprocal as the download cost.

% The main idea of the DSS paradigm is to replicate repositories over multiple non-colluding databases. A user can request parts of the desired data from each database, so that no single database can identify the contents retrieved by the user. 

Since the introduction of the PIR problem in \cite{chor1995private}, an extensive body of works have investigated efficient PIR schemes that yield either computational \cite{yekhanin2012locally,gasarch2004survey,song2000practical,ostrovsky2007survey} or information-theoretic privacy guarantees \cite{shah2014one,sun2017capacity,sun2017optimal,sun2019capacity,sun2018multiround,sun2018capacity,tajeddine2018private,banawan2018capacity,wang2017symmetric,freij2017private,sun2018private,jia2019x,lin2018mds,sun2019breaking,zhou2019capacity,banawan2018multi,zhang2017private,banawan2019capacity,tajeddine2018robust,banawan2018private,wang2018secure,wang2018capacity,tandon2017capacity,wei2018fundamental,shariatpanahi2018multi,Heidarzadeh_2018, Heidarzadeh_2019,attia2018capacity,tian2018capacity,wang2019symmetric,tajeddine2019private,yang2018private,jia2019cross,kumar2019achieving,raviv2018private,banawan2018private2,banawan2019capacity2}.
The former achieves privacy assuming  computational limitations at the DBs, where cryptographic assumptions are invoked to preserve privacy such that NP-hard computations are required to reveal the requested message identity.
 In information-theoretic PIR, the DBs are assumed to be computationally unbounded, thus achieving a higher level of assurance.  Perfect privacy is guaranteed if the queries do not reveal any information about the desired message (privacy) and the answers are sufficient to recover it (decodability).  An intuitive PIR solution  is to download all $K$ messages from a database. In fact, this is the only way to guarantee perfect privacy in the single database case. However, privacy comes at an impractical communication overhead. 

\textbf{Review of Recent Progress on Information-Theoretic PIR}: 
A practical way to increase the PIR capacity is to consider a distributed storage system (DSS) of $N$ databases. Shah. {\em et al.}  \cite{shah2014one} proposed a PIR scheme that achieves a rate of $1-\frac{1}{N}$ when $K$ messages are  replicated across $N$ non-colluding databases.  Later, Sun and Jafar \cite{sun2017capacity} characterized the PIR capacity for any $N$ and $K$ as $(1+1/N+1/N^{2}+\cdot \cdot \cdot+1/N^{K-1})^{-1}$. The original scheme introduced in \cite{sun2017capacity} achieves capacity when the message size $L$ is allowed to grow as a function of $N$ and $K$. Subsequently, they characterized the PIR capacity for a fixed message size \cite{sun2017optimal}. Since the appearance of the fundamental result of Sun and Jafar \cite{sun2017capacity}, numerous important and practically relevant variations of PIR have been considered.

Multi-round PIR allows multiple rounds of communication between the user and databases. While  interaction does not increase capacity, it can reduce the storage overhead at each database \cite{sun2018multiround}. Sun and Jafar \cite{sun2018capacity} considered the robust PIR problem where $M-N$ out of a total of $M>N$ databases fail to respond to  user queries. Additionally, they characterized the capacity when $T<N$ databases collude and  share the received queries. Tajeddine {\em et al.} \cite{tajeddine2018private} considered MDS-PIR  for coded databases  where each message is separately coded using an ($N,M$) MDS code.  Banawan and Ulukus \cite{banawan2018capacity} derived the coded PIR capacity  for arbitrary $N, M$, and $K$. Wang and Skogland \cite{wang2017symmetric} showed that the PIR capacity remains $1-\frac{M}{N}$ even if each message is coded. In \cite{freij2017private} and \cite{sun2018private}, the scenario of $N$ MDS-coded databases with $T$ colluding ones was presented. However, the capacity of this case is still an open problem (for other variants of MDS-PIR, see \cite{jia2019x,lin2018mds,sun2019breaking,zhou2019capacity}). In \cite{banawan2018multi} and \cite{zhang2017private}, the case of multi-message PIR, where the user can use one query to request more than one messages, was investigated. Banawan and Ulukus \cite{banawan2019capacity} characterized the PIR capacity with Byzantine databases where any subset of databases can be adversarial and respond untruthfully. In \cite{tajeddine2018robust}, Tajeddine {\em et al.} studied the same model but in the presence of colluding databases.  Banawan and Ulukus \cite{banawan2018private} studied PIR through a wiretap channel, where an eavesdropper tries to decode the content sent through the channel. Other variants of PIR in the presence of eavesdroppers are studied in \cite{wang2018secure},\cite{wang2018capacity}. 
 
 The problem of PIR was also studied when the user has a cache or side-information, which can be useful in increasing PIR capacity \cite{tandon2017capacity, wei2018fundamental,shariatpanahi2018multi, Heidarzadeh_2018, Heidarzadeh_2019}.
 PIR from storage-constrained databases was studied in  \cite{attia2018capacity,wei2019private,banawan2019private}, where capacity was characterized under the assumption of uncoded storage across databases. Recently, Tian {\em et al.} \cite{tian2018capacity} proposed a new capacity-achieving scheme with an optimal message size of $N-1$ and a minimum upload cost. 
 Other lines of work considered different privacy requirements from the original PIR model in \cite{sun2017capacity}. The problem of symmetric PIR (SPIR) was studied in \cite{sun2019capacity}, where the user must be able to retrieve the message of interest privately (user privacy), while at the same time the databases must avoid any information leakage  about the remaining $K-1$ messages (DB privacy).
The SPIR optimal download cost was characterized as $\frac{N}{N-1}$ with common randomness at least $\alpha = \frac{1}{N-1}$ bits per desired message bits. 
Latent-variable PIR was considered recently in \cite{samy2020latent}, where privacy is required for a latent variable describing a predefined user attribute. 
Additional interesting variants of PIR  can be found in \cite{wang2019symmetric,tajeddine2019private,yang2018private,jia2019cross,kumar2019achieving,raviv2018private,banawan2018private2,banawan2019capacity2}.

The novel coding schemes and fundamental ideas developed in the above works have also helped in advancing other problems beyond PIR. For instance, an interesting connection between blind interference alignment (BIA) and PIR was studied in \cite{sun2016blind} showing that a good BIA scheme translates to a good PIR protocol.
Secure and private distributed matrix multiplication has been considered in \cite{chang2018capacity,jia2019capacity,d2019gasp,chang2019upload,aliasgari2019distributed
} addressing the problem of computing a product of two matrices with some constraints on the identity of the product matrices and/or the information content in the matrices. Jia and Jafar \cite{jia2019x}  showed the connection of the secure and private distributed matrix multiplication to one variation of the MDS-PIR problem. Recently, the problem of private set intersection (PSI) was studied in \cite{wang2019private} from a PIR perspective and capacity results were obtained. 

%\ravi{Include connections of PIR to index coding, and secure distributed computing; Mohamed: see Tim's paper on upload-download cost of private and secure matrix multiplication, and a follow up work of Jafar which improves the result}. 
 
% \ravi{Move this paragraph on SPIR to previous section; this section should start with the line that the above works have all focused on perfect privacy (either for the user (in PIR problems) or for the message (as in SPIR formulations)).. 

\textbf{Relaxing Privacy Metrics for PIR}: The above works have all focused on perfect privacy, either for the user (as in PIR), or for both the user and the DBs (as in SPIR). The perfect privacy requirement usually comes at the expense of  high download cost and does not allow tuning the PIR efficiency and privacy according to the application requirements. In scenarios of frequent message retrieval, trading user or DB privacy for communication efficiency could be desirable. Ideally, one would select a desired leakage level and then design a leakage-constrained PIR scheme that guarantees such privacy while maximizing the PIR capacity.

A few previous works have introduced privacy definitions that relax the notion of perfect privacy.  Asonov {\em et al.} replaced privacy with the concept of repudiation \cite{asonov2002repudiative}. The repudiation property is achieved if some uncertainty remains about the desired message. However, this metric does not provide any information-theoretic privacy guarantees, as repudiation is satisfied even if the retrieved message can be identified with almost certainty.  
Recently, Toledo {\em et al.} \cite{toledo2016lower} adopted a game-based differential privacy definition to increase the PIR capacity at the expense of bounded privacy loss. However, their privacy definition only captures the privacy of the submitted queries. The authors propose several schemes that hide the query identity and study their cost. Although the query privacy can be thought of as functional equivalent to information-theoretic PIR in some cases, it does not satisfy the perfect privacy definition.
  
In our prior work, we introduced the Leaky PIR (L-PIR) where a bounded amount of leakage is allowed about the message identity \cite{samy2019capacity}. We adopted a concept similar to differential privacy  to bound the leakage as a function of a non-negative constant $\epsilon$. The leakage in privacy is achieved by constructing multiple biased ``retrieval paths'' across databases where each path realizes one query per database. Lin {\em et al.} \cite{lin2019weakly,lin2020capacity} relaxed user privacy by allowing bounded 
mutual information between the queries and the corresponding requested message index. Unlike \cite{lin2019weakly,lin2020capacity}, which deal with the average leakage (measured by mutual information), the L-PIR model in \cite{samy2019capacity} satisfies the privacy leakage constraints strictly for all possible  query/message index combinations, and thus provides stronger privacy guarantees. 

In another recent work, Guo {\em et al.} \cite{guo2019information} considered the problem of SPIR with perfect user privacy and relaxed DB privacy. DB privacy was relaxed by allowing a bounded mutual information (no more than $\delta$) between the undesired messages, the queries, and the answers received by the user. Similar to the original work on SPIR in \cite{sun2019capacity},  SPIR with relaxed DB privacy in \cite{guo2019information} requires sharing common randomness among DBs and comes at the expense of a loss in the PIR capacity.

{\bf Summary of contributions--} 
 We investigate a three-way tradeoff between  user privacy, DB privacy, and the communication efficiency of PIR.
We study the problem of \textit{Asymmetric Leaky PIR} (AL-PIR) where some information about the identity of the desired message is allowed to leak to the DBs, and some information about the undesired messages is allowed to leak to the user. The goal is to trade privacy in both directions for achieving gains in PIR capacity, thus making PIR more communication-efficient. 
For user privacy, we adopt the  metric introduced in our prior work \cite{samy2019capacity}, where the privacy bound is determined as a function of a non-negative constant $\epsilon$. For bounding DB privacy, we adopt a mutual information-based leakage metric %(studied in \cite{guo2019information})
to be  bounded by a non-negative constant $\delta$.
We next summarize the main contributions:

\begin{itemize}
\item We propose an AL-PIR scheme that satisfies  the leakage budgets in both directions for arbitrary values of $(\epsilon, \delta)$, an arbitrary number of $K$ messages,  and an arbitrary number of $N$ databases. The achievable download cost of this scheme is given by $D(\epsilon,\delta)= 1+\frac{1}{N-1}-\frac{\delta e^{\epsilon}}{N^{K-1}-1}$. This cost also represents an upper bound on the optimal download cost (lower bound on the capacity) of the AL-PIR.   We use an alternate perfect privacy PIR scheme that follows a \textit{path-based approach}, where a user's query is equivalent to selecting one of several possible paths across databases. %\textcolor{red}{LL: The idea of the path is not an established concept. Need to explain what a path is.} 
A path is defined as a set of queries, one per database,  that achieves decodability, however different paths incur different download costs. We leverage this cost imbalance to introduce leakage through the idea of biasing the path selection probabilities. A path giving a lower download cost can be used more frequently compared to higher download cost paths. This biasing introduces user privacy leakage. The path selection probabilities are chosen to minimize the download cost while satisfying the privacy budget, measured by $\epsilon$. To achieve DB privacy, our scheme requires sharing common randomness among databases. We combine the path-based approach  with the ideas of the scheme presented in \cite{guo2019information} to arrive at our general AL-PIR scheme. In particular, achieving a DB privacy leakage of no more than $\delta L$ bits, requires common randomness given by $\big(\frac{1}{N-1}-\frac{e^{\epsilon}+N^{K-1}-1}{N^{K-1}-1} \ \delta\big)L$ bits, which represents an upper bound on the optimal common randomness size.

\item We present a converse proof to obtain a lower bound on the optimum download cost (upper bound on capacity). This bound is characterized by $D^{*}(\epsilon,\delta)\geq 1+\frac{1}{Ne^{\epsilon}-1}-\frac{\delta}{(Ne^{\epsilon})^{K-1}-1}$.
The upper and lower bounds are shown to match each other at extreme values of epsilon ($\epsilon=0$; $\epsilon\to\infty$) and for any  $\delta$.
Moreover, we show through gap analysis that our upper and lower bounds are within a maximum multiplicative gap of $\frac{N-e^{-\epsilon}}{N-1}$ for any $\epsilon > 0$ and $\delta\geq 0$. 

\item We derive a lower bound on the optimal required common randomness at the databases.  This bound characterizes that achieving a DB privacy leakage of no more than $\delta L$ bits, requires shared randomness of size no less than $(\frac{1}{Ne^{\epsilon}-1} \ -\frac{(Ne^{\epsilon})^{K-1}}{(Ne^{\epsilon})^{K-1}-1} \delta)L$ bits.

\item We investigate the tradeoffs AL-PIR variations in both sides of leakage as special cases
of our general $(\epsilon, \delta)$ AL-PIR scheme.  In particular, we show a three-way tradeoff between download cost, user privacy, and DB privacy, such that enhancing one of them would be at the expense of the other two. We also show matching results for the following special cases for our derived bounds on the AL-PIR model: a) perfect user privacy (original PIR) \cite{sun2017capacity}, b) perfect user and DB privacy (SPIR) \cite{sun2019capacity}, c) Leaky user privacy (L-PIR) \cite{samy2019capacity}, and d) perfect user privacy and leaky DB privacy \cite{guo2019information}.
\end{itemize}

\begin{figure}[t]
\begin{center}
\includegraphics[width=0.9\linewidth]{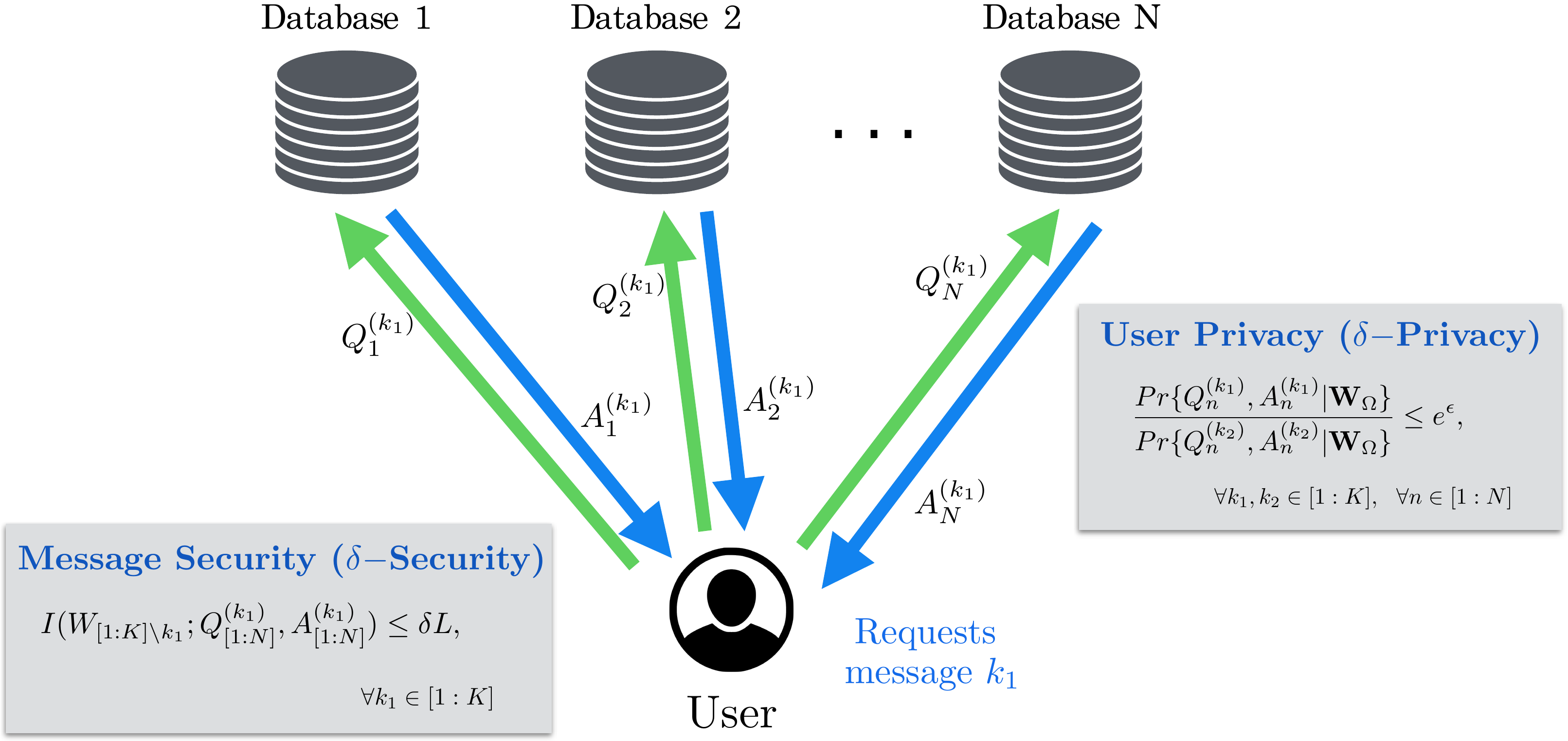}
\end{center}
\vspace{-0.2in}
\caption{Asymmetric leaky private information retrieval (AL-PIR) problem.}
\label{fig:sys}
\vspace{-0.2in}
\end{figure}

\vspace{-18pt}
\section{System Model: Asymmetric Leaky PIR}
\label{PF}
\vspace{-0.1in}
 We study the PIR problem illustrated in Figure \ref{fig:sys}. We consider $N$ databases  DB$_1$, DB$_2$,$\ldots$, DB$_N$  and $K$ independent messages $W_{1},W_{2},\ldots,W_{K}$, each of size $L$ bits, such that 
 \begin{equation}
 \label{eq:indepp}
 H(W_{1},W_2,\cdots,W_{K}) =\sum_{k=1}^{K} H(W_{k}),
 \end{equation}
  \begin{equation}
  \label{eq:size}
H(W_{1}) =H(W_{2})= \cdots = H(W_{K})=L.
 \end{equation}
 A user interested in privately retrieving $W_k$, $k \in [1:K]$\footnote{{\bf Notation:}  Through this work, we use the notation $[1:X]$ to represent the set of integers from $1$ to $X$.}
 sends $N$ separate queries $Q^{(k)}_{1},\cdots,Q^{(k)}_{N}$ to each of the $N$ DBs, where $Q^{(k)}_{n}$ denotes the query sent to the $n$th database (DB$_n$), $n\in[1:N]$, when retrieving message $W_k$. 
The $N$ DBs are assumed to be replicated and non-colluding, i.e., they store all the $K$ messages and they do not share the queries received from the user.
We also assume the DBs are interested in achieving privacy, i.e., the user must only decode the requested message subject to a leakage constraint.
To achieve DB privacy, the $N$ DBs are allowed to share common randomness denoted by a random variable $S$ of size $\alpha L$ bits, i.e., $H(S)=\alpha L$. Moreover, $S$ is not known to the user. 

Upon receiving $Q^{(k)}_{n}$, the  $n$th database generates the corresponding answer $A^{(k)}_{n}$ as a deterministic function of the query $Q^{(k)}_{n}$, the $K$ messages, and the shared common randomness $S$, i.e., 
\begin{equation}\label{answerfunction}
H\left(A^{(k)}_{n}
 |Q^{(k)}_{n},W_{1},\ldots,W_{K},S\right) = 0.
\end{equation}
The user must be able to decode the desired message $W_k$ upon receiving the answers  from the $N$ databases.  
Formally, the AL-PIR scheme must satisfy the following correctness, user privacy, and DB privacy constraints. 
\medskip

{\bf Correctness:}   Given queries $Q^{(k)}_{[1:N]}\triangleq \{Q^{(k)}_{1},\cdots,Q^{(k)}_{N}\}$, the user must be able to decode the desired message $W_{k}$, with  probability of error $P_e$, by collecting  the corresponding answers  $A^{(k)}_{[1:N]}\triangleq\{A^{(k)}_{1},\cdots,A^{(k)}_{N}\}$  from the $N$ DBs, i.e., 
\begin{equation}
\label{eq:corre}
H\left(W_{k}|Q^{(k)}_{[1:N]},A^{(k)}_{[1:N]}\right)=o(L)L,
\end{equation}
where $o(L)$ is any function that approaches zero as $L\rightarrow \infty$. $o(L)$ is set to zero if $P_e$ is required to be exactly zero.

{\bf $\delta-$DB privacy:}
In the original SPIR formulation \cite{sun2019capacity}, the authors assume no leakage to the user about the undesired messages. For a desired message $W_k$, perfect DB privacy is satisfied if
\begin{equation}
    I\left(\W_{[1:K]\setminus k};Q^{(k)}_{[1:N]},A^{(k)}_{[1:N]}\right)  = 0, \quad \forall k \in [1:K],
\end{equation}
where  $\W_{[1:K]\setminus k} \triangleq (W_1, \dots, W_{k-1},W_{k+1},\dots,W_K)$ is the set of all messages except $W_k$.
In this work, we relax this condition by assuming a general leaky DB privacy constraint.
The leaked information about the undesired messages must be bounded as,
\begin{equation}
\label{eq:MPri}
    I(\W_{[1:K]\setminus k};Q^{(k)}_{[1:N]},A^{(k)}_{[1:N]})\leq \delta L, \quad \forall k \in [1:K],
\end{equation}
where $\delta\geq 0$ is  a non-negative constant.

{\bf $\epsilon-$user privacy:} 
Under perfect user privacy, the privacy constraints are expressed as,
\begin{align}
\label{eq:DP0}
(A^{(k_1)}_{n},Q^{(k_1)}_{n},W_{1},\cdots,W_{K}) \sim (A^{(k_2)}_{n},Q^{(k_2)}_{n},W_{1},\cdots,W_{K}), \quad \forall k_1,k_2 \in [1:K].
 \end{align}
 This guarantees that the submitted queries are always independent of the message index. The previous constraint can be alternatively expressed as,
 \begin{align}
\label{eq:DP0'}
(A^{(k_1)}_{n},Q^{(k_1)}_{n}|\W_{\Omega}) \sim (A^{(k_2)}_{n},Q^{(k_2)}_{n}|\W_{\Omega}), \quad \forall k_1,k_2 \in [1:K],
 \end{align}
where $\W_{\Omega}$ is any subset of the $K$ messages, i.e.,  $\W_{\Omega}\subseteq \{W_1,\ldots,W_K\}$.
In this work,  the privacy constraint is relaxed such that given any subset $\mathbf{W}_{\Omega}$ of the $K$ messages,  the following likelihood ratio is bounded as follows: 

\begin{equation}
%\nonumber
\label{eq:DP}
\frac{Pr\{Q^{(k_1)}_{n}=\pi,A^{(k_1)}_{n}=\gamma| \mathbf{W}_{\Omega}\}}{Pr\{Q^{(k_2)}_{n}=\pi,A^{(k_2)}_{n}=\gamma| \mathbf{W}_{\Omega}\}} \leq e^{\epsilon}, \quad \forall k_1,k_2 \in [1:K],~~\forall n \in [1:N],
 \end{equation}
where $\pi$ and $\gamma$ represent any possible realizations for the queries and answers, respectively and $\epsilon$ is a non-negative constant. Unlike  perfect user privacy constraint which ensures  that queries and answers are independent of the message index,  the leaky privacy definition allows some queries and answers to be used more frequently when certain messages are retrieved. By setting $\epsilon=0$, the $\epsilon-$user privacy definition in \eqref{eq:DP} becomes equivalent to the perfect privacy constraint in \eqref{eq:DP0'}.

\medskip

\noindent \textbf{Other leaky user privacy definitions:} 
%\textcolor{red}{This paragraph does not explain or motivate in any way the asymmetric nature of the definition. If all you want to do is compare with other weaker  metrics for user privacy do so and leave the DB privacy out. If the intend is to explain why the asymmetry, do that, don't just state it .}
To  relax DB privacy, we adopt the mutual information metric  in \cite{guo2019information}.
On the other hand,  we use the probability metric we introduced in \cite{samy2019capacity} to bound the leakage of user privacy. The latter metric strictly satisfies the privacy  constraint for all possible query/message index combinations. 
We note that there are other weaker metrics one can use for relaxing user privacy.  In \cite{lin2019weakly}, Lin {\it et al.}  proposed a metric $\eta$ that gives a bound on the average privacy leakage over all databases for a desired message index given by a random variable $\theta\in[1:K]$ such that,
\begin{equation}
\begin{split}
   \frac{1}{N} \sum_{n=1}^N I(\theta;Q^{(\theta)}_n)\leq \eta.
    \end{split}
\end{equation}
% \ravi{How about the following?
% \begin{align}
%     I(\theta; Q_n^{(\theta)})\leq \rho
% \end{align}
% where $\theta$ is the random variable denoting the requested message index. It may be uniform, or many not be.. 
% }
Jia  {\it et al.} introduced the following privacy constraint \cite{jia2019capacity},
\begin{equation}
    H(A_n^{(k+1)}|W_1,\dots,W_k)-H(A_n^{(k)}|W_1,\dots,W_k)=\rho L, \quad \forall k \in [1:K], \quad \forall n \in [1:N],
\end{equation}
where   parameter $\rho$  controls the leakage budget, with $0\leq \rho \leq \frac{1}{N}$.  
In contrast to our $\epsilon$-user privacy definition in \eqref{eq:DP}, both of the  metrics provide average  privacy guarantees, i.e., they bound the average privacy leakage over all possible retrieval schemes. This means that the privacy leakage is allowed to exceed the required bound in the case of individual message retrievals.  In this work, we extend the definition in \eqref{eq:DP0} to  investigate the scenario when the distribution of the sent queries and the corresponding answers is allowed to depend on the requested message index within predefined limits. Also, the  AL-PIR model satisfies the $\epsilon-$user privacy definition strictly over all possible realizations of answers and queries. This ensures that leakage is always within the allowed budget $\epsilon$ for all individual message retrievals.

\medskip
\noindent {\bf Communication Cost:} To evaluate the performance of the AL-PIR scheme, 
%we consider the total amount of communication between the user and the DBs for retrieving the desired message. Similar to prior works, 
we adopt the Shannon theoretic formulation where the message size is assumed to be arbitrarily long and therefore, the upload cost is negligible compared to the download cost \cite{sun2017capacity}. In this case, the AL-PIR rate is the reciprocal of the download cost $D(\epsilon,\delta)$, which characterizes the total information bits  the user has to download to retrieve one desired message bit.
Let $D_{\epsilon,\delta}$ be the total number of downloaded bits to retrieve message $W_k$, for some  $\epsilon$ and $\delta$, and $L$ be the size of the desired message. The normalized download cost is given by,
\begin{equation}
D(\epsilon,\delta)=\frac{{D_{\epsilon,\delta}}}{L}=\frac{\sum_{n} H(A^{(k)}_{n})}{H(W_{k})}.
\end{equation}

We say that the pair $(L,D_{\epsilon,\delta})$ is achievable if there exists an AL-PIR scheme that satisfies the correctness, DB privacy, and user privacy  conditions in \eqref{eq:corre}, \eqref{eq:MPri}, and \eqref{eq:DP}, respectively, and can retrieve a message of size $L$ bits by downloading a total of $D_{\epsilon,\delta}$ bits.
 Our goal is to find the optimal download cost $D^{*}(\epsilon,\delta)$ such that
\begin{equation}
    D^{*}(\epsilon,\delta)=\min\{D_{\epsilon,\delta}/L:(L,D_{\epsilon,\delta}) \  \text{is achievable}\}.
    \end{equation} 
The capacity of the AL-PIR $C^{*}(\epsilon,\delta)$ is the reciprocal of $D^{*}(\epsilon,\delta)$, 
\begin{equation}
C^{*}(\epsilon,\delta)=\max\{L/D_{\epsilon,\delta}:(L,D_{\epsilon,\delta}) \  \text{is achievable}\}.
\end{equation}

\noindent {\bf Optimal common randomness size:} We are also interested in characterizing the fundamental limits of common randomness $S$  needed to be stored at the databases. In general, the common randomness size $\alpha$ is a function of the privacy budget parameters $(\epsilon,\delta)$. Therefore, in the following discussion, we use the notation $H(S)=\alpha(\epsilon,\delta)L$. We define $\alpha^*(\epsilon,\delta)$ as the minimum common randomness size that satisfies the correctness, DB privacy, and user privacy  conditions in \eqref{eq:corre}, \eqref{eq:MPri}, and \eqref{eq:DP}, respectively, i.e.,
\begin{align}
\label{eq:alpha-opt}
    \alpha^{*}(\epsilon,\delta) = \min\{\alpha(\epsilon,\delta): \text{\eqref{eq:corre}, \eqref{eq:MPri}, and \eqref{eq:DP} are satisfied}\}.
\end{align}

%The relaxed privacy definition shows us that any pair of answer and query can be biased towards a desired message.  

\section{Main Results and Discussion}
\label{MR}

%\begin{figure}[t]
%\begin{center}
% \setlength{\tabcolsep}{-0.037in}
% \begin{tabular}{cc}
%\includegraphics[trim=0.3in 0in 0.0in 0in, width=2.82in]{PIR1.eps} ~&\includegraphics[trim=0.3in 0in 0.0in 0in, width=2.82in]{PIR2.eps} \\
%(a) & (b) 
%\end{tabular}  
%\end{center} 
%\caption{(a) Lower and upper bounds of AL-PIR for $N=2$, and $K=2$. The download cost of perfect privacy is obtained when $\epsilon =0$, and (b) Upper bound of AL-PIR and the download cost of perfect privacy as $K$ grows.}
%\label{fig:UL}
%\end{figure}

In this section, we present our main results on the optimal download cost and the  required amount of shared randomness for AL-PIR. Given desired privacy budgets $\epsilon$ and $\delta$ for the user and DB privacy leakage, respectively, we state our main results in the following Theorems.

%\begin{theorem}
%\label{Thm:Rndm}
%To satisfy the DB privacy definitions, the required amount of shared randomness is
%\begin{equation}
%\label{eq:Rndm}
%    H(S)=\alpha L= L(\frac{1}{N-1}-\frac{e^{\epsilon}+N^{K-1}-1}{N^{K-1}-1} \ \delta).
%\end{equation}
%\end{theorem}

\begin{theorem}
\label{Thm:UB}
Define $d_1(\epsilon,\delta):= 1+\frac{1}{N-1}-\frac{\delta e^{\epsilon}}{N^{K-1}-1}$. For $N\geq 2$
and  shared randomness $S$ with  size $H(S)\geq \alpha_{1}(\epsilon,\delta) L,$ where
\begin{equation}
\label{eq:alpha-min}
   \alpha_{1}(\epsilon,\delta)= %\max\left(0,\
   %   \frac{1}{N-1}-\frac{e^{\epsilon}+N^{K-1}-1}{N^{K-1}-1} \ \delta\right)=   
   \begin{cases}
      \frac{1}{N-1}-\frac{e^{\epsilon}+N^{K-1}-1}{N^{K-1}-1} \ \delta, & 0\leq \delta < \delta_{1}(\epsilon),\\
      0, &  \delta > \delta_{1}(\epsilon),
      \end{cases}
\end{equation}
the optimal download cost of  AL-PIR, satisfying both the $\epsilon-$user privacy and $\delta-$DB privacy definitions, is upper-bounded  by
\begin{equation}
\label{eq:UB}
    D^{*}(\epsilon,\delta)\leq D^{\text{UB}}(\epsilon,\delta)= \begin{cases} d_1(\epsilon,\delta), & 0\leq \delta < \delta_{1}(\epsilon),\\
    d_1(\epsilon,\delta_{1}(\epsilon)), & \delta\geq\delta_{1}(\epsilon).
    \end{cases}
\end{equation}
In \eqref{eq:alpha-min} and \eqref{eq:UB},  $\delta_{1}(\epsilon)$ is the maximum DB privacy leakage (when no common randomness is required, i.e., $ \alpha_{1}(\epsilon,\delta)=0$) which is a function of the allowed  user privacy leakage $\epsilon$, and is given by,
\begin{align}
\label{eq:delta-max-lb}
\delta_{1}(\epsilon)=\frac{N^{K-1}-1}{(N-1)(e^{\epsilon}+N^{K-1}-1)}.
\end{align}
% Using this value of $\delta_{1}(\epsilon)$, the expression for  $d_1(\epsilon,\delta_{1}(\epsilon))$ is given as
% \begin{align}
% \label{eq:d1_e-d}
%     d_1(\epsilon,\delta_{1}(\epsilon))= 1+\frac{N^{K-1}}{e^{\epsilon}+N^{K-1}-1}\Bigg (\frac{1}{N}+\dots+\frac{1}{N^{K-1}} \Bigg ).
% \end{align}
\end{theorem}
The  proof of Theorem~\ref{Thm:UB} is presented in Section~\ref{Pr1}. As a result of Theorem~\ref{Thm:UB}, we have the following remark.

\begin{remark}
The required size of shared randomness for our achievability scheme, as given by $\alpha_{1}(\epsilon,\delta)$ in \eqref{eq:alpha-min}, yields an upper bound on the optimal size of minimum shared randomness $\alpha^*(\epsilon,\delta)$ as defined in \eqref{eq:alpha-opt}, i.e., 
$
   \alpha^*(\epsilon,\delta)\leq \alpha_1(\epsilon,\delta).$ %= \max\left(0,\
     % \frac{1}{N-1}-\frac{e^{\epsilon}+N^{K-1}-1}{N^{K-1}-1} \ \delta\right).
%\end{equation}
Moreover, $\alpha_1(\epsilon,\delta)$ is also sufficient to satisfy  $(\epsilon', \delta')$ privacy constraints, such that $\epsilon'\geq \epsilon$ and $\delta'\geq\delta$. In other words, if a given amount of common randomness is sufficient to satisfy $(\epsilon, \delta)$ privacy, then it is also sufficient if the privacy budgets are increased. 
\end{remark}

%\begin{remark}[\textbf{Bounds for $\delta\geq\delta_{1}(\epsilon)$}]
%\label{remark2}
%Substituting any  $\delta>\delta_{1}(\epsilon)$ in \eqref{eq:alpha-min} yields $ \alpha_{1}(\epsilon,\delta)=0$. This means that no shared randomness is required. For such $\delta$, the download cost is fixed to
%\begin{align}
 %  D^{\text{UB}}(\epsilon,\delta)=d_1(\epsilon,\delta_{1}(\epsilon))&= 1 +\delta_{1}(\epsilon)= 1+\frac{N^{K-1}}{e^{\epsilon}+N^{K-1}-1}\Bigg (\frac{1}{N}+\dots+\frac{1}{N^{K-1}} \Bigg ),\quad\forall\delta\geq\delta_{1}(\epsilon). \label{eq:d1_e-d}
%\end{align}
%\end{remark}

In Figure~\ref{fig:delta-var}, we show the effect of $\epsilon$ and $\delta$ on the download cost for the case when $N=K=2$. We can observe the following: 
a) the download cost is a monotonically decreasing function of the privacy budgets $\epsilon$ and $\delta$;
b) as $\epsilon $ approaches infinity, which corresponds to no user privacy, the achieved download cost approaches $1$; c) for $\epsilon =0$ (perfect user privacy) and as $\delta$ approaches zero (perfect DB privacy), the achieved download cost is $2$ which matches the case of SPIR studied in \cite{sun2019capacity} where the optimal download cost is $\frac{N}{N-1}=2$; and d) for $\delta \geq \delta_{1}(\epsilon)=\nicefrac{1}{(e^{\epsilon}+1)}$ (or $\epsilon\geq \ln(\nicefrac{1}{\delta}-1)$), the download cost is only a function of $\epsilon$ (the line corresponding to $\delta=0.4$).

% {\color{red} \ravi{This bullet is very confusing to read; why does it approach 1 ?? what is $\delta_2$; }
% d) the download cost approaches $1$ according to the relation in \eqref{eq:d1_e-d} when $\delta > \delta_{1}$, i.e., DB privacy leakage is maximum, as defined in \eqref{eq:delta-max}. If we define $\epsilon_{\delta}$ as the inverse function of \eqref{eq:delta-max}
% \begin{align}
% \epsilon_{\delta} = \ln \left((N^{K-1}-1)\left(\frac{1}{\delta(N-1)}-1\right)\right),
% \end{align}
% then when $\epsilon > \epsilon_{\delta}$,
% %the trade-off is given by the relation in \eqref{eq:d1_e-d}.
%  we obtain the same download cost vs. user privacy tradeoff shown in \eqref{eq:d1_e-d}. For instance, when $\epsilon \geq \epsilon_{\delta_2}$ the two tradeoff curves coincide for both $\delta = 0.4$ and $\delta = 0.04$. 
%  }

\begin{figure}[t]
\begin{center}
\includegraphics[width=0.65\linewidth]{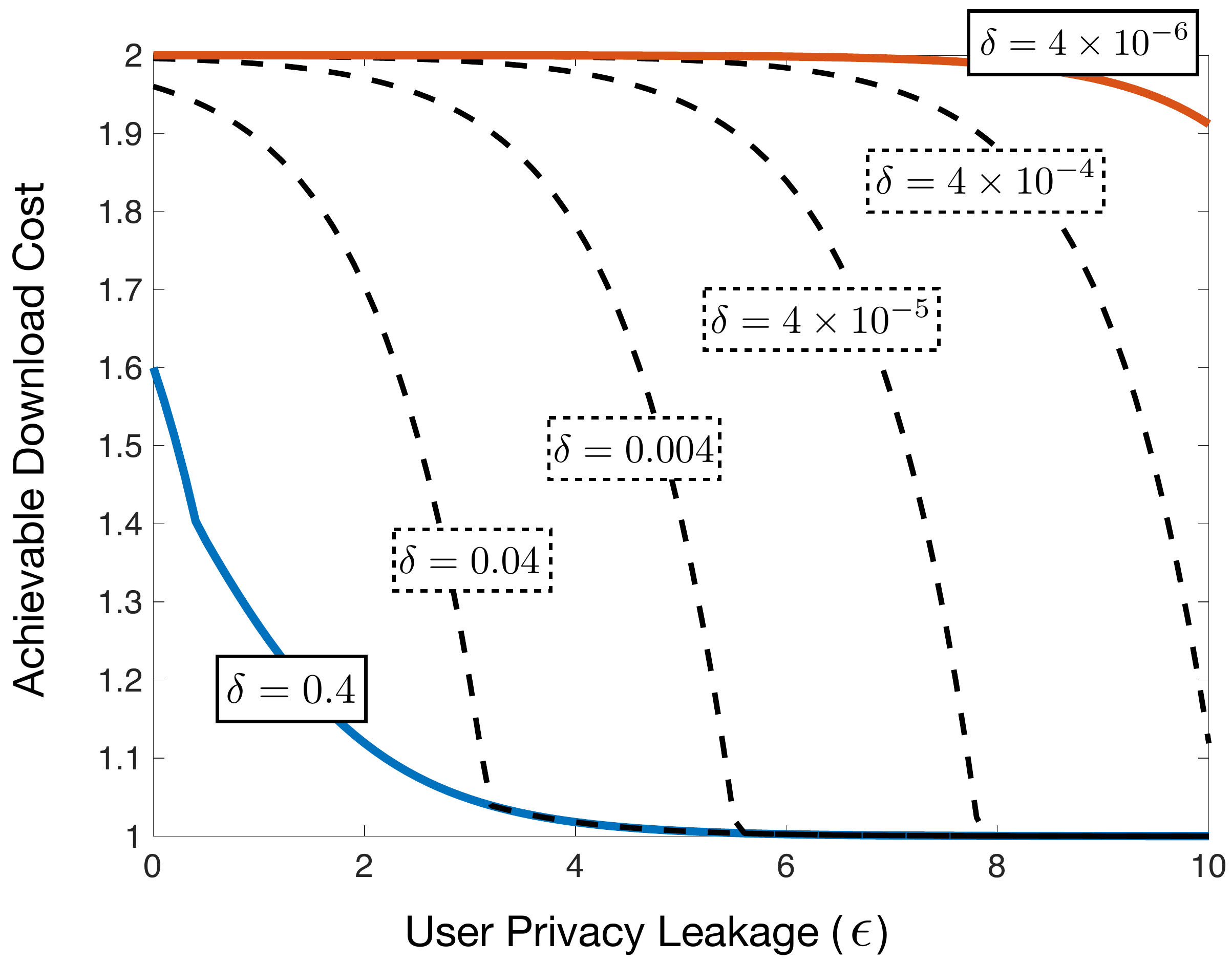}
\end{center}
%\vspace{-0.1in}
\caption{The achievable download cost for our AL-PIR scheme when $N=K=2$  as  a function of $\epsilon$ for different values of $\delta$.}
\label{fig:delta-var}
%\vspace{-0.1in}
\end{figure}

%\textcolor{red}{I know we have worked in the two theorems but the presentation is inconsistent. Theorem 1 presents the common randomness first, followed by the download cost followed be $delta_1$. Theorem 2 presents the common randomness last. Shall we unify to one style?}

\vspace{10pt}
\begin{theorem}
\label{Thm:LB}
Define $d_2(\epsilon,\delta):= 1+\frac{1}{Ne^{\epsilon}-1}-\frac{\delta}{(Ne^{\epsilon})^{K-1}-1}$.
For $N\geq 2$, and  shared randomness $S$ with  size $H(S)\geq \alpha_{2}(\epsilon,\delta) L,$ where
\begin{align}
     \alpha_2(\epsilon,\delta)=% \max\left(0,\  \frac{1}{Ne^{\epsilon}-1} \ -\frac{(Ne^{\epsilon})^{K-1}}{(Ne^{\epsilon})^{K-1}-1} \delta\right).
       \begin{cases}
     \frac{1}{Ne^{\epsilon}-1} \ -\frac{(Ne^{\epsilon})^{K-1}}{(Ne^{\epsilon})^{K-1}-1} \ \delta, & 0\leq \delta < \delta_{2}(\epsilon),\\
      0, &  \delta > \delta_{2}(\epsilon),
      \end{cases}
\end{align}
the optimal download cost of AL-PIR subject to $\epsilon-$user privacy and $\delta-$DB privacy  is lower-bounded by
\begin{equation}
\label{eq:LB}
    D^{*}(\epsilon,\delta)\geq D^{LB}(\epsilon,\delta)= \begin{cases} d_2(\epsilon,\delta), &0\leq \delta < \delta_{2}(\epsilon),\\
    d_2(\epsilon,\delta_{2}(\epsilon)), & \delta\geq\delta_{2}(\epsilon),
    \end{cases}
\end{equation}
where 
\begin{align}
    &\delta_{2}(\epsilon)=\frac{(Ne^{\epsilon})^{K-1}-1}{(Ne^{\epsilon}-1)(Ne^{\epsilon})^{K-1}}.    \label{eq:delta-max-ub}%\\
    %&d_2(\epsilon,\delta_{2}(\epsilon))=1+\delta_{2}(\epsilon)=1+ \frac{1}{Ne^{\epsilon}}+\dots+\frac{1}{(Ne^{\epsilon})^{K-1}}.\label{eq:d2_e-d}
\end{align}
Furthermore, the optimal size of common randomness satisfying $\epsilon-$user privacy and $\delta-$DB privacy is lower-bounded by $\alpha^*(\epsilon,\delta) \geq \alpha_2(\epsilon,\delta)$.
%\begin{align}
%\label{eq:LB-alpha}
 %   \alpha^*(\epsilon,\delta) \geq \alpha_2(\epsilon,\delta)=% \max\left(0,\  \frac{1}{Ne^{\epsilon}-1} \ -\frac{(Ne^{\epsilon})^{K-1}}{(Ne^{\epsilon})^{K-1}-1} \delta\right).
   %    \begin{cases}
    % \frac{1}{Ne^{\epsilon}-1} \ -\frac{(Ne^{\epsilon})^{K-1}}{(Ne^{\epsilon})^{K-1}-1} \ \delta, & 0\leq \delta < \delta_{2}(\epsilon),\\
 %    5 0, &  \delta > %\delta_{2}(\epsilon).
      %\end{cases}
%\end{align}

\end{theorem}

%{\color{red} \noindent (OLD TEXT) Next, we show the optimal download cost for the case of $N=1$ database.
%\begin{remark}
%Achieving any level of DB privacy is not feasible for the case of one DB, $N=1$. This is because it is impossible to simultaneously satisfy user privacy and correctness, along with DB privacy. From \eqref{eq:DP}, for any finite $\epsilon$, we notice that any query/answer pair has to be requested to retrieve each of the $K$ messages with non-zero probability. This implies that to guarantee satisfying the correctness property we have to include the $K$ messages in all possible answers. Consequently,  we cannot provide any level of DB privacy. 
%\end{remark}
%Interestingly, we note that even if we ignore the DB privacy, 
%the proposed relaxation of user privacy does not help in reducing the download cost 
%for a single database scenario.  We prove proposition \ref{prob:N=1} in Appendix \ref{pr3}.
%\begin{proposition}
%\label{prob:N=1}
%Let we ignore any requirements of DB privacy, i.e., $\delta\to \infty$. For $N=1$, the optimal download cost of an AL-PIR scheme, satisfying the $\epsilon-$user privacy definition for any finite $\epsilon$,  matches the download cost under perfect privacy conditions,
%\begin{equation}
 %   D^{*}(\epsilon, \delta \to\infty)=K. 
%\end{equation}
%\end{proposition}
%}

%{\color{blue} \noindent (NEW TEXT; do not delete; need to discuss) 
The proof of Theorem~\ref{Thm:LB} is presented in Section~\ref{Pr2}. We note that the results in Theorems \ref{Thm:UB} and \ref{Thm:LB} hold for $N\geq 2$ DBs. In the following proposition, we characterize the capacity for the case of one database.
\begin{proposition}
\label{prob:N=1}
The optimal download cost $D^{*}(\epsilon, \delta)$ for $N=1$ and for any $0\leq \epsilon<\infty$ is given by:
\begin{align}
    D^{*}(\epsilon, \delta)=
    \begin{cases}
\infty, \quad \delta < (K-1),\\
K, \quad \delta = (K-1).\\
\end{cases}
\end{align}
\end{proposition}
The above result shows that the problem of AL-PIR for one database is degenerate. In particular, to satisfy the $\epsilon$-user privacy constraint,  any query/answer pair has to be requested to retrieve each of the $K$ messages with non-zero probability. Since $N=1$, the only solution is to download all messages, i.e., a download cost of $K$. However, upon downloading all $K$ messages, the leakage about the remaining $(K-1)$ messages is fixed and given by $\delta= K-1$. Hence, if the DB privacy budget is $\delta< (K-1)$, the AL-PIR problem is infeasible and the capacity is $0$, i.e., $D^*(\epsilon,\delta<K-1)=\infty$. We prove Proposition \ref{prob:N=1} in Appendix \ref{pr3}. 

%\textcolor{red}{What is the point of a delta greater than K-1? Should the second branch in 22 be simply equality to K-1?} \textcolor{blue}{Islam: the parameter $\delta$ can take any positive value. Although, it is clear that the maximum possible leakage is K-1, I think we should keep it like that}

In the next Corollary, we show that our proposed scheme in Theorem~\ref{Thm:UB} is information-theoretically optimal for perfect user privacy, i.e., $\epsilon = 0$, and is optimal within a maximum multiplicative gap ratio of $\frac{N-e^{-\epsilon}}{N-1}$ for any  $(\epsilon,\delta)$. The proof of the corollary is presented in Appendix  \ref{app_gap}.

\begin{corollary}
\label{Thm:Gap}
The multiplicative gap ratio between the upper and  lower bounds on the download cost of the AL-PIR, given by Theorems \ref{Thm:UB} and \ref{Thm:LB}, respectively, is bounded as follows:
\begin{align}
\label{eq:Gap}
    \frac{D^{\text{UB}}(\epsilon,\delta)}{D^{\text{LB}}(\epsilon,\delta)}\leq
\frac{N-e^{-\epsilon}}{N-1}.
\end{align}
\end{corollary}

In Figure~\ref{fig:UL_Bounds}, we show the upper and lower bounds on the download cost of the AL-PIR and the numerical multiplicative gap ratio, as a function of system parameters $(N,K,\epsilon, \delta)$.
Specifically, in Figure~\ref{fig:UL_Bounds}a, we set the allowed DB privacy leakage to the maximum leakage, i.e., $\delta >\max( \delta_{1}(\epsilon),\delta_{2}(\epsilon))$ as defined in \eqref{eq:delta-max-lb} (no shared randomness required for this case). This gives the results of the  L-PIR model considered in \cite{samy2019capacity}.  
As the number of messages increases, both upper and lower bounds increase, whereas both  decrease with $N$. This happens as increasing $N$ increases the number of bits that can be utilized as a side information to retrieve the desired message. On the other hand, increasing $K$ adds an overhead on any retrieval scheme to satisfy the privacy by considering the symmetry among downloaded bits from different messages.
We observe a similar trend for the multiplicative gap ratio as well. In Figure~\ref{fig:UL_Bounds}b, we fix the value of the DB privacy leakage to $\delta = 4\times 10^{-5}$. This choice  insures that $\delta < \min(\delta_{1}(\epsilon),\delta_{2}(\epsilon))$ for all $\epsilon \in[0:10]$ considered in the plots.  We note that while increasing $K$ does not have significant impact on the bounds, both the download cost and multiplicative gap ratio decrease with $N$. Moreover, we observe that the bounds match when $\epsilon =0$, i.e., when perfect user privacy is required, and when $\epsilon \to\infty$, i.e., no user privacy is required. 
% In the following remark, we show the relation between $\delta_1(\epsilon)$ and $\delta_1(\epsilon)$.
%  \begin{remark}
%  For any $\epsilon$, we have $\delta_1(\epsilon)\geq \delta_2(\epsilon)$. From Theorem \ref{Thm:UB} and \ref{Thm:LB}, we notice that we can express $d_1(\epsilon,\delta_1(\epsilon))$ as $d_1(\epsilon,\delta_1(\epsilon))=1+\delta_1(\epsilon)$. Similarly, we have $d_2(\epsilon,\delta_2(\epsilon))=1+\delta_2(\epsilon)$. For any $\delta\geq\max\left(\delta_1(\epsilon),\delta_2(\epsilon)\right)$, $D^*(\epsilon,\delta)$ can be bounded as follows
%  \begin{equation}
%  \label{eq:delta1>2}
%   1+\delta_2(\epsilon)=d_2(\epsilon,\delta_2(\epsilon))=D^{LB}(\epsilon,\delta)\leq D^*(\epsilon,\delta)\leq D^{UB}(\epsilon,\delta)=  d_1(\epsilon,\delta_1(\epsilon))=1+\delta_1(\epsilon).
%  \end{equation}
%  Following \eqref{eq:delta1>2}, $\delta_1(\epsilon)$ is always greater than or equal $\delta_2(\epsilon)$.
%  \end{remark}

\begin{figure*}[t]
\begin{center}
\includegraphics[width=0.95\linewidth]{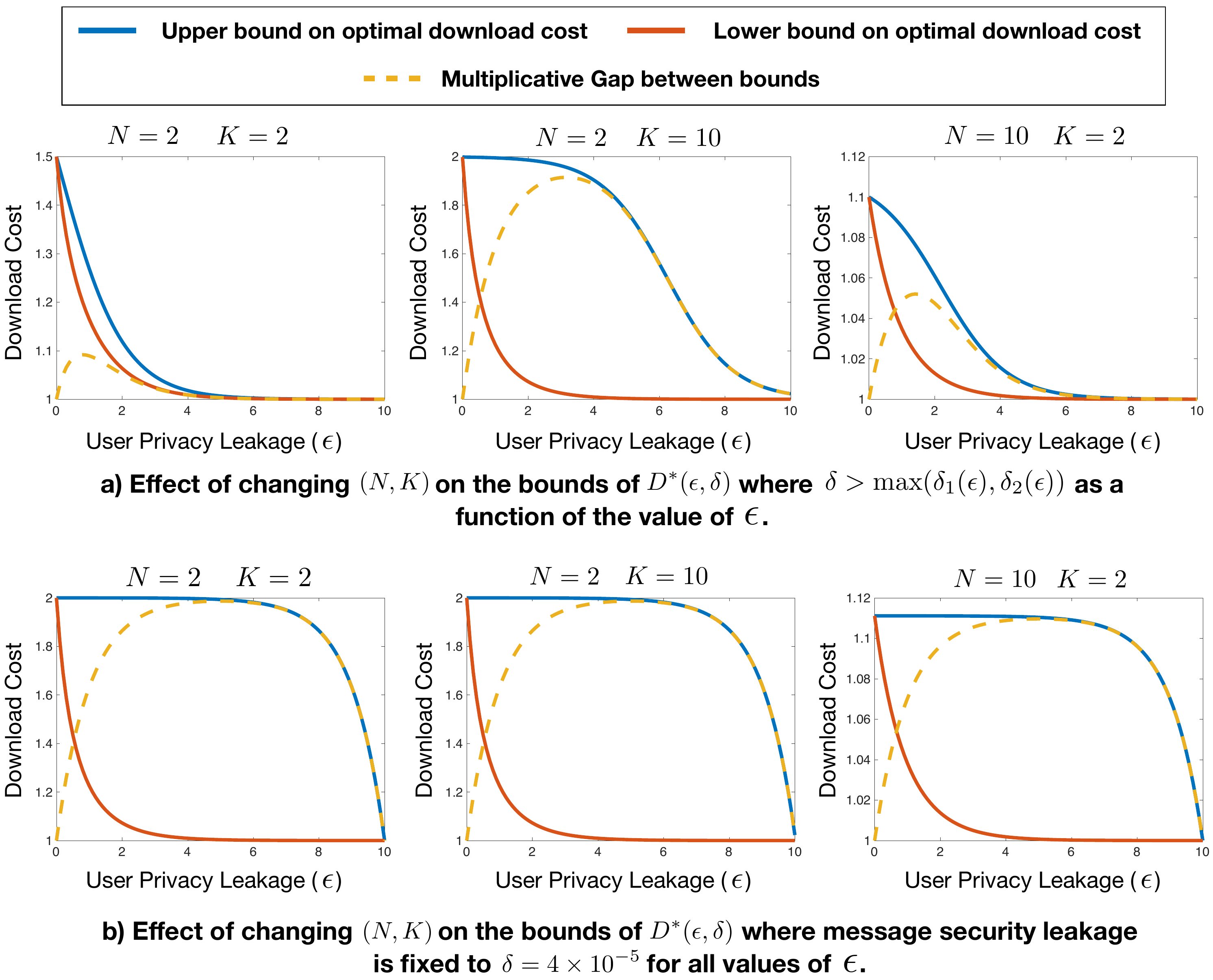}
\end{center}
\vspace{-0.15in}
\caption{ Lower and upper bounds of AL-PIR for different values of $N$, $K$ and $\delta$ as $\epsilon$ increases.} 
\label{fig:UL_Bounds}
\vspace{-0.15in}
\end{figure*}

The generality of the AL-PIR problem formulation allows us to recover several existing results on PIR as special cases of  Theorems \ref{Thm:UB} and \ref{Thm:LB}. These  cases are discussed in the following remark.

\begin{remark}[\textbf{Connections to state-of-the-art results}]
From Theorems \ref{Thm:UB} and \ref{Thm:LB}, the lower and upper bounds on the optimal download cost $D^{*}(\epsilon,\delta)$ for any  $(\epsilon,\delta)$ can be used to derive the following prior results.

\noindent $\bullet$ \hspace{5pt} \textbf{No user privacy and perfect DB privacy} $(\epsilon \to\infty,\delta = \delta_{1}(\epsilon \to\infty)=\delta_{2}(\epsilon \to\infty)= 0)$.
From the shared randomness bounds \eqref{eq:alpha-min} and \eqref{eq:delta-max-lb}, when $\epsilon \to \infty$ and $\delta =0$, we get that $ \alpha^*(\epsilon \to\infty,0)=\alpha_{1}(\epsilon \to\infty,0) =\alpha_{2}(\epsilon \to\infty,0)=0$, i.e., no shared randomness is needed.
Substituting the $\epsilon$ and $\delta$ values in the download cost bounds \eqref{eq:UB} and \eqref{eq:LB}, we get $D^{\text{UB}}(\epsilon\to\infty,\delta  =0) = D^{\text{LB}}(\epsilon\to\infty,\delta  =0)=1$, meaning that
the upper and lower bounds are matching  and give an optimal download cost of $D^{*}(\epsilon\to\infty,\delta  =0) = 1$. That is, AL-PIR is achieved by only downloading the requested file from any of the databases. 
%\textcolor{red}{LL:You need to be more precise. Do you use both Theorems everywhere? If the bounds coincide you should state so. }

\noindent $\bullet$ \hspace{5pt} \textbf{Perfect user privacy and maximum leakage on DB privacy \cite{sun2017capacity}} $(\epsilon = 0, \delta = \delta_{1}(\epsilon =0)= \delta_{2}(\epsilon =0)= \frac{N^{K-1}-1}{N^{K-1}(N-1)})$.
%We obtain the results of the original PIR formulation in \cite{sun2017capacity} using Theorems \ref{Thm:UB} and \ref{Thm:LB}.
We obtain the original PIR result in \cite{sun2017capacity} for perfect user privacy leakage $\epsilon=0$.
For this special case, we get the optimal required shared randomness characterized by $\alpha^*(\epsilon,\delta)= \alpha_{1}(\epsilon,\delta) =\alpha_{2}(\epsilon,\delta)= 0$, i.e., no shared randomness is needed. Using the bounds in \eqref{eq:UB} and \eqref{eq:LB}, we obtain matching upper and lower bounds, giving an optimal download cost of
\begin{align}
    D^{*}(\epsilon=0,\delta = \frac{N^{K-1}-1}{N^{K-1}(N-1)}) &=D^{\text{LB}}(\epsilon=0,\delta = \frac{N^{K-1}-1}{N^{K-1}(N-1)})=D^{\text{UB}}(\epsilon=0,\delta = \frac{N^{K-1}-1}{N^{K-1}(N-1)})\nonumber\\&=  
    1+ \frac{1}{N}+\dots+\frac{1}{N^{K-1}}.
\end{align}

 %\textcolor{red}{LL:You need to be more precise. Do you use both Theorems everywhere? If the bounds coincide you should state so. }

\noindent $\bullet$ \hspace{5pt} \textbf{Perfect user privacy and DB privacy \cite{sun2019capacity}} $(\epsilon = 0,\delta = 0)$. By setting $\epsilon = 0$, $\delta  =0$ in Theorems \ref{Thm:UB} and \ref{Thm:LB}, we obtain the SPIR results in \cite{sun2019capacity} where the optimal required shared randomness is given by $ \alpha^*(0,0)=\alpha_{1}(0,0) = \alpha_{2}(0,0)=\frac{1}{N-1}$ and the optimal download cost is obtained using the bounds in \eqref{eq:UB} and \eqref{eq:LB} as
\begin{align}
    D^{*}(\epsilon=0,\delta = 0)=D^{\text{LB}}(\epsilon=0,\delta = 0)=D^{\text{UB}}(\epsilon=0,\delta = 0)=
    1+ \frac{1}{N-1}.
\end{align}

\noindent $\bullet$ \hspace{5pt} \textbf{Leaky user privacy and maximum leakage on DB privacy \cite{samy2019capacity}}. %$(\epsilon,\delta = \delta_{1} = \frac{N^{K-1}-1}{(N-1)(e^{\epsilon}+N^{K-1}-1)})$.
We obtain the L-PIR results in \cite{samy2019capacity} for any level of user privacy leakage $\epsilon$ and a DB privacy leakage  $\delta\geq\max\left(\delta_{1}(\epsilon),\delta_{2}(\epsilon)\right)$, where the optimal required shared randomness is given by $\alpha^*(\epsilon,\delta)= \alpha_{1}(\epsilon,\delta) =\alpha_{2}(\epsilon,\delta) = 0$ and the bounds on the optimal download cost are obtained using \eqref{eq:UB} and \eqref{eq:LB} as
\begin{align}
    &D^{*}(\epsilon,\delta)\geq D^{\text{LB}}(\epsilon, \delta_{2}(\epsilon))=
1+ \frac{1}{Ne^{\epsilon}}+\dots+\frac{1}{(Ne^{\epsilon})^{K-1}}, \nonumber\\
    &D^{*}(\epsilon,\delta )\leq D^{\text{UB}}(\epsilon, \delta_{1}(\epsilon))=
    1+ \frac{N^{K-1}-1}{(N-1)(e^{\epsilon}+N^{K-1}-1)}.
\end{align}

\noindent $\bullet$ \hspace{5pt} \textbf{Perfect user privacy and Leaky DB privacy \cite{guo2019information}} $(\epsilon=0,\delta)$. 
For perfect user privacy $\epsilon=0$ and DB privacy leakage characterized by $\delta$, we obtain the results in \cite{guo2019information}, where the optimal required shared randomness is characterized by $ \alpha^*(0,\delta)=\alpha_{1}(0,\delta) =\alpha_{2}(0,\delta)= \frac{1}{N-1} + \frac{N^{K-1}}{N^{K-1}-1}\delta$ and the optimal download cost is obtained using the bounds in \eqref{eq:UB} and \eqref{eq:LB} as
\begin{align}
    D^{*}(\epsilon=0,\delta)=D^{\text{LB}}(\epsilon=0,\delta)=D^{\text{UB}}(\epsilon=0,\delta)=
    \frac{N}{N-1}- \frac{\delta}{N^{K-1}-1}.
\end{align}

\end{remark}

     %Figure~\ref{fig:UL}(a) shows the upper and lower bounds of AL-PIR for $N=K=2$. One can indeed observe the decrease of the download cost (increase in capacity) with $\epsilon$, quickly approaching the minimum possible value of 1. Fig.~\ref{fig:UL}(b) compares the download cost of perfect PIR with AL-PIR when $\epsilon=1.$ 

%\section{L-PIR Construction via path biasing with maximum DB Privacy Leakage}
%\label{L-PIR}

\section{Proof of Theorem \ref{Thm:UB} : Upper Bound on  $D^*(\epsilon,\delta)$ for the AL-PIR}
\label{Pr1}

%Before describing the AL-PIR scheme, we revisit the L-PIR scheme introduced originally in our previous work \cite{samy2019capacity} for user privacy leakage and maximum DB Privacy leakage.
The leakage in user privacy is achieved using the path-based approach introduced  in our previous work \cite{samy2019capacity}.  A retrieval path is equivalent to a set of queries across databases that guarantee decodability. Possible retrieval paths have different download costs. The probability of selecting each path is chosen to minimize the download cost while satisfying the privacy budget, measured by $\epsilon$, which is a process referred to as {\em path biasing.}
First, we give the following example for $N=K=2$ to describe the idea of path biasing to achieve $\epsilon-$user privacy leakage with DB Privacy leakage $(\delta \geq \delta_{1}(\epsilon))$.

\subsection{AL-PIR Example for $N=2$, $K=2$, and privacy leakage $(\epsilon, \delta \geq \delta_{1}(\epsilon))$}

 Consider the simplest non-trivial PIR setting with $N=2$ DBs and $K=2$ messages denoted by $W_1$ and $W_2$. To motivate the construction of AL-PIR, we first recall the perfect PIR scheme proposed by Sun and Jafar in \cite{sun2017capacity}. Assume that the messages $W_1=\{a_1,\dots,a_4\}$ and $W_2=\{b_1,\dots,b_4\}$, are each $L=4$ bits long.
  Figure~\ref{fig:ex1} shows a retrieval structure for $W_1$ using the scheme in \cite{sun2017capacity}. The main idea is that one can use coding and leverage side information from the other database to reduce the download cost to $3/2$. We highlight that the shown bit indices represent one possible permutation of the real indices. Thus, $W_1$ retrieval can be obtained through multiple bit structures that are selected uniformly and have an equal download cost of $3/2$. 

In Figure~\ref{fig:ex2}, we  show an alternative PIR scheme in which the requested message can be downloaded via sequences of structures that give unequal download cost. In particular, when the user wants to retrieve message $W_1$, it picks one of the four possible queries/paths: 
\begin{itemize}
\item Path $\mathcal{P}_1$:($\emptyset,W_1$): Send no request to DB$_1$ and request $W_1$ from DB$_2$. This  path/query has a download cost of $L$ bits. 
\item Path $\mathcal{P}_2$:($W_1,\emptyset)$: Request $W_1$ from DB$_1$ and send no request to DB$_2$. This path has a download cost of $L$ bits. 
\item Path $\mathcal{P}_3$:($W_2,W_1\oplus W_2)$: Request $W_2$ from DB$_1$ and $W_1 \oplus W_2$ from DB$_2$. This path has a download cost of $2L$ bits. 
\item Path $\mathcal{P}_4$:($W_1\oplus W_2,W_2)$: Request $W_1\oplus W_2$ from DB$_1$ and $W_2$ from DB$_2$. This path has a download cost of $2L$ bits. 
\end{itemize}
Paths $\mathcal{P}_1$ and $\mathcal{P}_2$, which have lower download cost, are selected with probability $p$, whereas  higher download cost paths $\mathcal{P}_3$ and $\mathcal{P}_4$ are selected with probability $q$. From the total probability theorem, we have
\begin{equation}
\label{eq:pq11}
    2p+2q=1.
\end{equation}

The answer of DB$_n$ can take four different structures, $\pi_{n,1},\dots,\pi_{n,4}$. These structures represent the element addition of all possible subsets of $\{W_1,W_2\}$. 
 Note that the selection probability of any structure $\pi_{n,j}$, $j\in[1:4]$ equals the selection probability of all paths containing that structure. Also, there is one path per message that contains each structure $\pi_{n,j}$. For example, $\pi_{1,2}=\{W_1\}$ is paired with $\pi_{2,2}=\{\emptyset\}$ to retrieve $W_1$, or it can be paired with $\pi_{2,3}=\{W_1\oplus W_2\}$ for $W_2$ retrieval.   Let the path selection probabilities be uniform, i.e., $p=q=\frac{1}{4}$. Thus, each structure is selected with probability $\frac{1}{4}$, irrespective of the requested message index. It is straightforward to show that this probability assignment satisfies the perfect privacy definition in \eqref{eq:DP0}. Moreover, although the cost varies per path, the uniform path selection yields an optimal average download cost of $3/2$. Therefore, this path-based PIR scheme is also optimal and matches the result of Sun and Jafar \cite{sun2017capacity} for perfect privacy.

\begin{figure}[t]
\begin{center}
\includegraphics[width=0.4\linewidth]{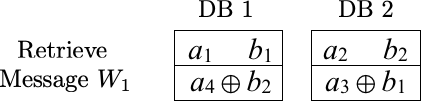}
\end{center}
%\vspace{-0.1in}
\caption{The original PIR scheme in \cite{sun2017capacity} for $N=2, K=2,$ and $L=4.$}
\label{fig:ex1}
%\vspace{-0.1in}
\end{figure}
\begin{figure}[t]
\begin{center}
\includegraphics[width=0.8\linewidth]{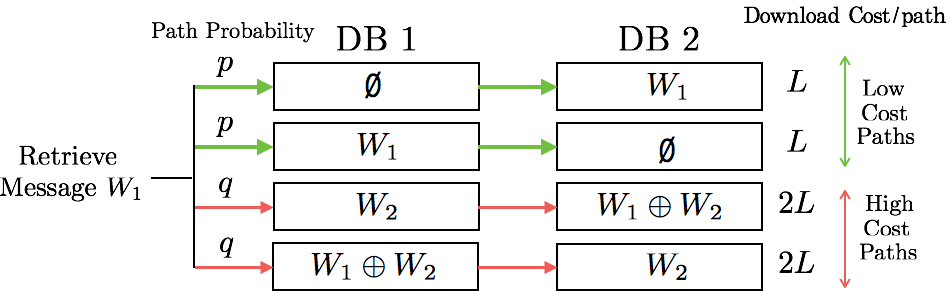}
\end{center}
%\vspace{-0.1in}
\caption{AL-PIR scheme for $N=2, K=2$, general $\epsilon$, and $\delta\geq\delta_{1}(\epsilon)$.}
\label{fig:ex2}
%\vspace{-0.2in}
\end{figure}

 {\bf Improving the download cost via path biasing (achieving $\epsilon-$user privacy).} The leaky privacy definition in \eqref{eq:DP} together with the path-based scheme described above, lead us to consider schemes that bias the path selection process for retrieving desired messages. We next show that this helps reduce the average download cost for any non-zero $\epsilon$. Intuitively, if we assign higher selection probability to paths with  lower download cost than the average (for example $L$), an overall lower cost can be achieved at the expense of some bounded loss of privacy due to the biasing. 
The question we pose is whether there are values $p \neq q$ that yield an average download cost less than $\frac{3}{2}$ and simultaneously satisfy the $\epsilon-$user privacy definition in \eqref{eq:DP}. The probability $Pr\{Q^{(i)}_{n}=\pi,A^{(i)}_{n}=\gamma| \mathbf{W}_{\Omega}\}$ can be expressed as
\begin{equation}
\label{eq:20}
    Pr\{Q^{(i)}_{n}=\pi,A^{(i)}_{n}=\gamma| \mathbf{W}_{\Omega}\}=Pr\{Q^{(i)}_{n}=\pi| \mathbf{W}_{\Omega}\}Pr\{A^{(i)}_{n}=\gamma|Q^{(i)}_{n}=\pi, \mathbf{W}_{\Omega}\}.
\end{equation}
The term $Pr\{Q^{(i)}_{n}=\pi| \mathbf{W}_{\Omega}\}$ depends on the path selection probability. To provide privacy, for any answer to a specific structure  $\pi$, the term $Pr\{A^{(i)}_{n}=\gamma|Q^{(i)}_{n}=\pi, \mathbf{W}_{\Omega}\}$ should be constant independently of the requested message. 
To meet the privacy definition in \eqref{eq:DP}, it is sufficient to show that the possible structures to each query satisfy:
\begin{equation}
\label{eq:DP_p-ex}
  \frac{\Pr(\pi_{n,j}|i=1)}{\Pr(\pi_{n,j}|i=2)} <e^{\epsilon}, \quad \forall n \in \{1,2\},  \ j\in [1:4],
\end{equation}
where $\Pr(\pi_{n,j}|i=k),$ is the probability of retrieving structure $\pi_{n,j}$ when the desired message is $k$.  Based on the scheme in Figure \ref{fig:ex2}, there are two cases for each structure $\pi_{n,j}$: \begin{enumerate}[(i)]
    \item $\pi_{n,j}$ is  used to recover $W_1$ and $W_2$ with the same probability either $p$ or $q$, then
\begin{equation}
 \frac{P(\pi_{n,j}|i=1)}{P(\pi_{n,j}|i=2)} = 1,   
\end{equation}
 which clearly satisfies \eqref{eq:DP}.
 \item  $\pi_{n,j}$ is selected with different probabilities $p$ and $q$ to retrieve $W_1$ and $W_2$, respectively, and vice versa. Then, $p$ and $q$ must satisfy
    \begin{equation}
    \label{eq:11}
        e^{-\epsilon}\leq \frac{\Pr(\pi_{n,j}|i=1)}{\Pr(\pi_{n,j}|i=2)} = \frac{p}{q} \leq  e^{\epsilon}.
    \end{equation}
\end{enumerate}  
Invoking the fact that the sum of path probabilities must equal one, we use \eqref{eq:pq11} to  substitute by  $q=0.5-p$ and rewrite  \eqref{eq:11}  as
\begin{equation}
    \frac{p}{0.5-p}\leq e^{\epsilon}.
\end{equation}
This gives us the following inequality,
 \begin{equation}
\label{eq:DP_p2}
    p \leq \frac{e^{\epsilon}}{2(1+e^{\epsilon})}.
\end{equation}
Therefore, we can pick $p$ that satisfies \eqref{eq:DP_p2} with equality, and then select  $q=0.5 -p$, as a valid choice of path selection probabilities which satisfy the $\epsilon-$user privacy constraint. 
 
{\bf Computing the download cost $D(\epsilon,\delta\geq \delta(\epsilon))$.}
Since our scheme is symmetric with respect to messages, the same download cost is obtained for the retrieval of message $W_1$ or message $W_2$. Then, the average download cost can be written as
\begin{equation}
    D(\epsilon,\delta\geq \delta(\epsilon))=\frac{\sum_{j=1}^{4}\Pr\{\mathcal{P}=\mathcal{P}_j\}\cdot D_{\mathcal{P}_j}}{ L},
\end{equation}
where $\Pr\{\mathcal{P}=\mathcal{P}_j\}\in\{p,q\}$ is the probability that path $\mathcal{P}_j$ is chosen and $D_{\mathcal{P}_j}$ is the cost of path $\mathcal{P}_j$.
From Figure \ref{fig:ex2}, we know that $D_{\mathcal{P}_1}=D_{\mathcal{P}_2} = L,$ and $D_{\mathcal{P}_3}=D_{\mathcal{P}_4}=2L.$ Hence, $D(\epsilon,\delta\geq \delta(\epsilon))$ equals
 \begin{align}
 \label{eq:Dcost_p1_p2}
      D(\epsilon,\delta\geq \delta(\epsilon))&=\frac{2\times p\times L +2 \times q\times (2L)}{L}\nonumber\\
      &=2p+4q\nonumber\\
      &\overset{(a)}{=}2-2p\nonumber\\
      &\overset{(b)}{\geq} 2-\frac{e^{\epsilon}}{(1+e^{\epsilon})}, 
       \end{align}
where $(a)$ follows from \eqref{eq:pq11}, and $(b)$ follows from \eqref{eq:DP_p2}. Hence, the download cost of this scheme (when $p=\nicefrac{e^{\epsilon}}{2(1+e^{\epsilon})}$),  can be rewritten as 
 \begin{equation}
 \label{eq:Dcost_eps}
 \begin{split}
     D^*(\epsilon,\delta\geq\delta_1(\epsilon))=\frac{3}{2}-\frac{e^{\epsilon}-1}{2(e^{\epsilon}+1)},
 \end{split}
 \end{equation}
 which is lower than $\frac{3}{2}$, the optimal download cost under perfect privacy. Note that a lower cost cab be achieved for any $\epsilon.$
 
 {\bf Computing DB privacy leakage $\delta$.} We have shown in the above example that the biased selection probability of the path-based scheme can trade user privacy for lower download cost. We now calculate the DB privacy leakage. From the above leaky construction, we can show that
\begin{equation}
    H(A_{[1:2]}^{(1)})=D(\epsilon,\delta\geq\delta_1(\epsilon))\times L\geq\frac{3}{2} L-\frac{e^{\epsilon}-1}{2(e^{\epsilon}+1)} L.
\end{equation}
Similarly, the average size $H(A_{[1:2]}^{(1)}|W_2)$ of answers given $W_2$ is known and can be expressed as
\begin{align}
 H(A_{[1:2]}^{(1)}|W_2)&=\sum_{j=1}^{4}\Pr\{\mathcal{P}=\mathcal{P}_j\}\cdot D_{\mathcal{P}_j|W_2} \nonumber\\
 &=2\times p\times L +2 \times q\times L=L,
\end{align}
where $D_{\mathcal{P}_j|W_2}$ is the cost of path $\mathcal{P}_j$ when $W_2$ is given.
This makes the DB privacy leakage, or the information revealed about $W_2$, equal to
\begin{equation}
    I(W_2;A_{[1:2]}^{(1)})= H(A_{[1:2]}^{(1)})- H(A_{[1:2]}^{(1)}|W_2)\geq\frac{1}{2} L-\frac{e^{\epsilon}-1}{2(e^{\epsilon}+1)} L=\delta_1(\epsilon) L.
\end{equation}
We highlight that this construction can achieve a lower DB privacy leakage compared to the perfect privacy scenario in \cite{sun2017capacity} where $I(W_2;A_{[1:2]}^{(1)})=\nicefrac{L}{2}$, without the need for any shared randomness. However, this construction cannot fulfill the DB privacy constraint if $\delta<\delta_1(\epsilon)$. In the following example, we introduce a construction that can satisfy any DB privacy requirement with the utilization of the common randomness.
 %Here, we only consider user leakage and assume that any DB privacy leakage budget is tolerable (no certain $\delta$ is required). 

 \begin{figure}[t]
\begin{center}
\includegraphics[width=0.9\linewidth]{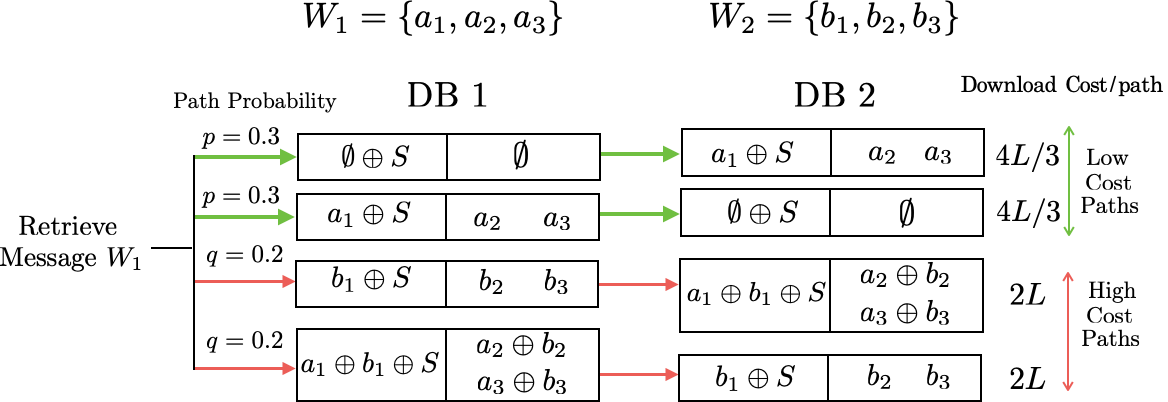}
\end{center}
%\vspace{-0.1in}
\caption{AL-PIR scheme for $N=2$, $K=2$, $\epsilon =\ln(1.5)$, and $\delta =\nicefrac{4}{15}$.}
\label{fig:ex3}
%\vspace{-0.1in}
\end{figure}

\subsection{AL-PIR example with $N=2$, $K=2$, $\epsilon=\ln{(1.5)}$, and $\delta=\nicefrac{4}{15}$}
Figure~\ref{fig:ex3} shows an example of a possible AL-PIR scheme with $N=2$, $K=2$, $\epsilon=\ln{(1.5)}$, and $\delta=\nicefrac{4}{15}$. We observe that the allowed DB privacy leakage $\delta$ is less than $\delta_1(\epsilon)$,
\begin{equation}
    \frac{4}{15}=\delta<\delta_1(\epsilon)=\frac{1}{2} -\frac{e^{\epsilon}-1}{2(e^{\epsilon}+1)}=0.4.
\end{equation}
Now, assume that each of the two messages is of size $L=3$ bits, $W_1=\{a_1,a_2,a_3\}$ and $W_2=\{b_1,b_2,b_3\}$. To  satisfy the $\delta-$DB privacy condition, we include the least required amount of shared randomness $S$ that has a size of $ \alpha_{1}(\epsilon,\delta)L$, where $ \alpha_{1}(\epsilon,\delta)$ is computed from \eqref{eq:alpha-min}:
\begin{equation}
    \alpha_{1}(\epsilon,\delta)L= (\frac{1}{N-1}-\frac{e^{\epsilon}+N^{K-1}-1}{N^{K-1}-1} \ \delta)L=\frac{L}{3}=1 \ \text{bit}.
\end{equation}
%$\nicefrac{L}{3}=1$ bit.
Each message  is divided into two parts as follows: $W_1$ is divided into  $W_1^{(1)}=\{a_1\}$ (size of $S=\nicefrac{L}{3}$), and $W_1^{(2)}=\{a_2,a_3\}$; and  
 $W_2$ is divided into  $W_2^{(1)}=\{b_1\}$, and $W_2^{(2)}=\{b_2,b_3\}$. 
 
 Suppose that the user wants to retrieve $W_1.$ The user can use any of the four possible paths shown in Figure~\ref{fig:ex2}, where a path is defined as a query set $Q^{(k)}_{[1:N]}$ which satisfies, together with its corresponding answer, the correctness and privacy constraints. %\textcolor{red}{Is this the right definition, or we have to define via the answers?}.  
 However, these paths have different download costs. The first two paths have a cost of $\nicefrac{4L}{3}$ bits, whereas the other two paths have a cost of $2L$ bits. The correctness of the scheme is straightforward, the XOR addition of the two structures forming each path results in getting $a_1$, $a_2$, and $a_3$.
 To reduce the download cost by trading user privacy, similar to the previous example, we select  the lower cost paths with probability $p=0.3$, whereas the higher cost paths are assigned  a probability $q=0.2$.
 These selection probabilities are chosen such that both $\epsilon-$user privacy and $\delta-$DB privacy conditions are satisfied. As we will discuss later in more details, the ratio describing $\epsilon-$user privacy leakage in \eqref{eq:DP} is given by the maximum ratio between the probabilities of selecting different paths, represented here as $p/q= 1.5 = e^{\epsilon}$.

  When $W_1$ is requested by the user, the DB privacy leakage is described as the information user can decode about $W_2$. We notice that the first two paths in Figure~\ref{fig:ex2} do not reveal any information about $W_2$, while using the other two paths, the user can decode the two bits $W_2^{(2)}= \{b_2,b_3\}$. This gives the average DB privacy leakage as 
  \begin{equation}
   I(W_2;A_{[1:2]}^{(1)})= 2\times0.3\times0 + 2\times0.2\times2= 0.8= 0.8\times\frac{L}{3}=\frac{4}{15}L=\delta L.   
  \end{equation}
Hence, this achieves the $\delta-$DB privacy condition. 
The average number of downloaded bits for this scheme is  
 \begin{equation}
     D_{\epsilon,\delta}= 2\times0.3\times4+2\times0.2\times 6 =4.8 \ \text{bits},
 \end{equation} which yields a download cost of \begin{equation}
     D\big(\epsilon=\ln{(1.5),\delta=\frac{4}{15}}\big)=\nicefrac{4.8}{3}=1.6.
 \end{equation} We highlight that this scheme clearly  improves the download cost in comparison to the perfect SPIR, which has a download cost of $\nicefrac{N}{N-1}=2$, at the expense of some loss in user and DB privacy.

\subsection{General $(\epsilon, \delta)$ AL-PIR Construction}

In this section, we generalize the AL-PIR scheme in the previous examples for arbitrary values of $N$, $K$, and asymmetric privacy leakage characterized by the pair $(\epsilon,\delta)$.
Assume there are $K\geq 2$ messages, $W_1,\dots, W_K$. 
Consider a random permutation of the databases indices.
Let each message $W_k$ be divided into two parts $W_k=\{W_{k}^{(1)}, W_{k}^{(2)}\}$ such that 
\begin{equation}
    H(W_{k}^{(1)})=(N-1) \alpha_{1}(\epsilon,\delta) L, \end{equation}
\begin{equation}
H(W_{k}^{(2)})=L-(N-1) \alpha_{1}(\epsilon,\delta) L,
\end{equation}
where $ \alpha_{1}(\epsilon,\delta)$
is the minimum required amount of shared randomness for the AL-PIR scheme to ensure the $\delta-$DB privacy and computed  as
\begin{align}
     \alpha_{1}(\epsilon,\delta)&=\max\left(0,\ \frac{1}{N-1}-\frac{e^{\epsilon}+N^{K-1}-1}{N^{K-1}-1} \ \delta\right)
     \nonumber\\
     &=\begin{cases}
      \frac{1}{N-1}-\frac{e^{\epsilon}+N^{K-1}-1}{N^{K-1}-1} \ \delta, & 0\leq \delta < \delta_{1}(\epsilon),\\
      0, &  \delta > \delta_{1}(\epsilon).
      \end{cases}
     \label{eq:alpha}
\end{align}
Furthermore, let each $W_{k}^{(1)}$ and $W_{k}^{(2)}$ be divided into $N-1$ equal sub-packets,
\begin{equation}
    W_{k}^{(1)}=\{W_{k,1}^{(1)}, \dots, W_{k,N-1}^{(1)}\},
\end{equation}
\begin{equation}
    W_{k}^{(2)}=\{W_{k,1}^{(2)}, \dots, W_{k,N-1}^{(2)}\},
\end{equation}
such that for all $\ell\in[1:N-1]$,
$
    H(W_{k,\ell}^{(1)})= \alpha_{1}(\epsilon,\delta) L,\quad
    H(W_{k,\ell}^{(2)})=(\frac{1}{N-1}- \alpha_{1}(\epsilon,\delta))L.
$
For instance, in the example of Figure~\ref{fig:ex3} where $ \alpha_{1}(\epsilon,\delta)=\nicefrac{1}{3}$ and $L=3$ bits,  $W_1$ is divided into $W_{1}^{(1)}=\{a_1\}$ of size 1 bit, and $W_{1}^{(2)}=\{a_2,a_3\}$ of size 2 bits.   
%Each message $W_k$ has a size of  $L=N-1$ bits, $W_k=\{W_{k,1}, \dots, W_{k,N-1}\}$. 

For a requested message $W_i$, the DBs mask $W_{k}^{(1)}$'s, $k\in[1:K]\setminus i$, with the secret key $S$. The content of $W_{k}^{(2)}$'s may be allowed to leak to the user.  To retrieve a required message $W_i$, the user first selects one of the possible retrieval paths across the $N$ DBs. Any path is formed by a set of $N$ queries, $Q^{(k)}_{[1:N]}$, which are submitted to the respective DBs. %\textcolor{red}{It becomes more and more obvious to me that a path is nothing more than a query. No? Why the new terminology} 
The selected path has to fulfill two requirements: (i) the path correctly recovers $W_i$; (ii) the $N$ submitted queries satisfy both the $\epsilon-$user privacy and $\delta-$DB privacy conditions. 
The user sends the following query vector  to  DB$_n$
\begin{equation}
\pi_{n,i}=(x_1,\ldots,x_{i-1},(x_{i}+n)_N,x_{i+1},\ldots, x_K),\quad x_k\in[0:N-1],\ k\in[1:K],  
\end{equation}
where $(x_{i}+n)_N$ denotes $(x_{i}+n) \ (\text{mod} \ N)$. This $K\times 1$ vector gives the indices of the $K$ message bits, one bit for each message,  that should be included in the answers. The design of $\pi_{n,i}$ makes sure that all submitted queries $\pi_{n,i}$'s include the same indices of all undesired messages, and different indices of the required $W_i$. Then, the identical undesired bits, within the $N$ collected answers, can be utilized to decode the desired bits.
After DB$_n$ receives the query $\pi_{n,i}$, it responds with answer $\gamma_{n}(\pi_{n,i})$,
{
\begin{align}
\label{eq:paths}
  &\gamma_{n}(\pi_{n,i})=\big\{\underset{k\in[1:K]\setminus i}{\bigoplus} W_{k,x_k}^{(1)}\oplus S \oplus W_{i,(x_i+n)_N}^{(1)}, \underset{k\in[1:K]\setminus i}{\bigoplus} W_{k,x_k}^{(2)}\oplus W_{i,(x_i+n)_N}^{(2)}\big\},
\end{align}}
\normalsize
where $\bigoplus$ represents the summation via XOR operation. We denote  by $W_{k,0}^{(1)}$ and $W_{k,0}^{(2)}$ the null or the empty set $\emptyset$. 
This response ensures protecting all $W^{(1)}_{k,\ell}$'s by encoding them with $S$. 
We observe that once one of the $N$ queries is designed, i.e., the indices $x_k$'s are chosen, the remaining $N-1$ queries are  deterministic functions of these chosen indices. As $x_k\in[0:N-1]$ for any $k\in[1:K]$, each  $\pi_{n,i}$ can be represented by $N^K$ different vectors. Each of these vectors creates one possible path to retrieve $W_i$. Thus, for a specific permutation of DBs indices, we have $N^K$ possible paths in general. 
For the example in Figure~\ref{fig:ex3}, there are $N^K=4$ paths for the retrieval of $W_1$. The first retrieval path is created from the queries $\pi_{1,1}=(0,0)$, and $\pi_{2,1}=(1,0)$ with corresponding answers  $\gamma_{1}(\pi_{1,1})=\{S,\emptyset\}$, and   $\gamma_{2}(\pi_{2,1})=\big\{a_1\oplus S,\{a_2,a_3\} \big\}$. 
A general form for one possible path  is shown in Figure~\ref{fig:op}.

{\bf Analysis of Correctness:}  The user can  decode the sub-messages $W_{i,\ell}^{(1)}$ and $W_{i,\ell}^{(2)},$  $\forall \ \ell\in[1:N-1]$, of the requested message ($W_{i,\ell}$)  using the information retrieved from DBs $N-x_i$ and $(N+\ell-x_i)_N$  as follows %\textcolor{red}{The notation is so convoluted here, it is impossible to understand and the construction bears no intuition.}
\begin{align}
      \{W_{i,\ell}^{(1)},W_{i,\ell}^{(2)}\}&=
      \gamma_{N-x_i}(\pi_{N-x_i,i}) \oplus \gamma_{N+\ell-x_i}(\pi_{(N+\ell-x_i)_N,i})\nonumber\\
      &=\big\{\underset{k\in[1:K]\setminus i}{\bigoplus} W_{k,x_k}^{(1)}\oplus S \oplus W_{i,0}^{(1)},\underset{k\in[1:K]\setminus i}{\bigoplus} W_{k,x_k}^{(2)} \oplus W_{i,0}^{(2)} \big\}\nonumber\\& \quad \oplus \big\{\underset{k\in[1:K]\setminus i}{\bigoplus} W_{k,x_k}^{(1)}\oplus S \oplus  W_{i,\ell}^{(1)},\underset{k\in[1:K]\setminus i}{\bigoplus} W_{k,x_k}^{(2)} \oplus  W_{i,\ell}^{(2)}\big\}\nonumber\\
       &=  \big\{W_{i,0}^{(1)}\oplus W_{i,\ell}^{(1)}, \ W_{i,0}^{(2)}\oplus W_{i,\ell}^{(2)}\big\}\nonumber\\&=\big\{W_{i,\ell}^{(1)}, \ W_{i,\ell}^{(2)}\big\}.
\end{align}

% Similar to the previous construction, we can divide the paths into two categories according to the values $x_k$'s, $k\neq i$, and their download costs.
% The selection probabilities are also $p$ and $q$, where $p\geq q$, for structures that initiate paths of cost $(1+\alpha)L$ and $\frac{N}{N-1}L$, respectively. 

{\bf Proof of $\epsilon-$user privacy:} We note that the total download cost for each path is not fixed, but it depends on the choice of $x_k$'s, $k\in[1:K]$. Then, we have two types of paths:
\begin{itemize}
    \item {\textbf{Lower cost paths} ($\mathbf{\ \forall k\in[1:K]\setminus i, \ x_k= 0}$):}\\
    Generally, $N$ possible paths belong to this case, those created from queries $\pi_{n,i}$ where $x_i=0,1,\dots,N-1$. In this case, we have
    \begin{equation}
     \gamma_{N-x_i}(\pi_{N-x_i,i})= \big\{\underset{k\in[1:K]\setminus i}{\bigoplus} W_{k,0}^{(1)}\oplus S \oplus W_{i,0}^{(1)},\underset{k\in[1:K]\setminus i}{\bigoplus} W_{k,0}^{(2)} \oplus W_{i,0}^{(2)} \big\}=\big\{S,\emptyset\big\},   
    \end{equation}
     i.e., we only download the secret key $S$ of size $ \alpha_{1}(\epsilon,\delta)L$ from DB$_{N-x_i}$. Whereas, for other databases, all structures download data of the form $\big\{W_{i,\ell}^{(1)}\oplus S,W_{i,\ell}^{(2)}\big\}$, each  of size $\frac{L}{N-1}$ bits. In total for these type of paths, the user needs to download $(N-1) \frac{L}{N-1}+ \alpha_{1}(\epsilon,\delta)L =(1+ \alpha_{1}(\epsilon,\delta))L$ bits from all DBs. %All paths comes through these structures have a cost of $L$ bits. 
    \item {\textbf{Higher cost paths} ($  \mathbf{\exists \  k\in[1:K]\setminus i, x_k\neq 0}$):}\\
    For this case, there are $N^K-N$ possible paths. 
   Here,  all requested query structures are of size $\frac{L}{N-1}$ bits, and  the user needs to download $\frac{N}{N-1}L$  bits in total from all DBs. 
\end{itemize}
%We observe from the construction that the user has $N^K$ choices of structures, per database, to select among them.  For each structure, there are $N-1\times N-2\times \dots\times 1= N-1!$ possible choices for  the remaining $N-1$ requested structures. Thus, for each requested message, there are $N-1!$ paths that can pass through any structure. These paths have the same cost, either $L$ or $N\frac{L}{N-1}$, as they share the same $x_j$'s.  Hence, for each requested message, we have: 
%\begin{enumerate}[(i)]
 %   \item $N\cdot (N-1)!$ paths of cost L that are initiated from a structure that has all $x_j$'s equal zero.
  %  \item $(N^K-N)\cdot (N-1)!$ paths of cost $N\frac{L}{N-1}$ that are initiated from any of the $N^K-N$ structures that has  $x_j\neq 0$ for at least one $j\neq i$.
%\end{enumerate} 

Without loss of generality, we assign probabilities $p$ and $q$, where $p\geq q$, to higher cost and lower cost paths, respectively\footnote{Due to symmetry, paths belonging to the same type are assigned the same probability. Assigning different probabilities does not improve the download cost or the privacy.}, such that
\begin{equation}
\label{eq:sum_prob_N}
 N\times p+(N^K-N)\times q=1.   
\end{equation}
%Taking into account the possible $N-1!$ paths coming through each structure, each structure is picked with a probability of either $p$ or $q$.
%Figure~\ref{fig:ex3} shows the AL-PIR scheme for $K=3$.
%\vspace{-0.1in}
\begin{figure}[t]
\begin{center}
\includegraphics[width=0.85\linewidth]{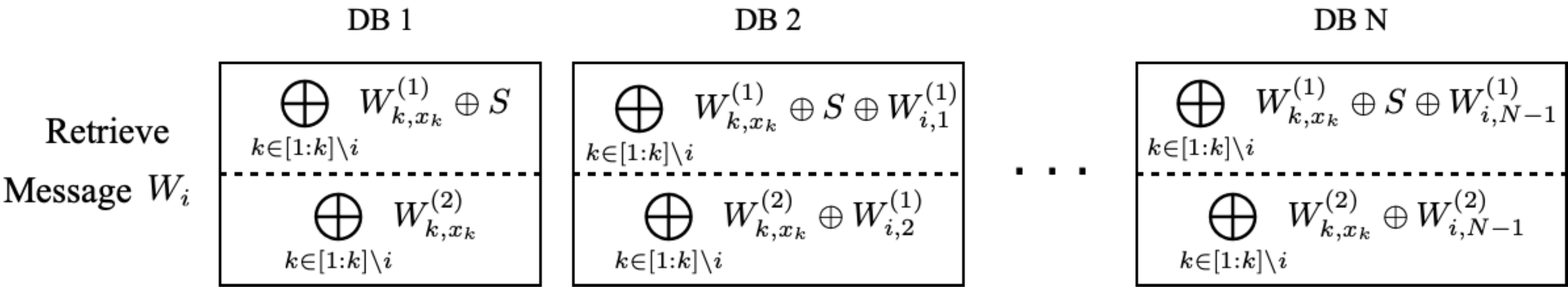}
\end{center}
%\vspace{-0.1in}
\caption{One path of the AL-PIR scheme that retrieves $W_i$ with $x_i=N-1$.}% \textcolor{red}{Also change the queries to much your initial notation.}}
\label{fig:op}
%\vspace{-0.1in}
\end{figure}

We can see that any query $\pi_{n,i}$ with  certain $x_k$'s can be used to recover any desired message. This is obtained by requesting that query with other $N-1$ queries that share the same $x_k$'s of the $K-1$ remaining messages. This is crucial to satisfy the $\epsilon-$user privacy requirements because accessing a structure does not eliminate any of the message possibilities. Furthermore, each structure is selected to retrieve $W_i$ with the same probability of selecting the path coming through it, either $p$ or $q$. Similar to \eqref{eq:20}, $\Pr\{Q^{(i)}_{n}=\pi,A^{(i)}_{n}=\gamma| \mathbf{W}_{\Omega}\}$ can be expressed as
\begin{equation}
\label{eq:20'}
    \Pr\{Q^{(i)}_{n}=\pi,A^{(i)}_{n}=\gamma| \mathbf{W}_{\Omega}\}=\Pr\{Q^{(i)}_{n}=\pi| \mathbf{W}_{\Omega}\}\Pr\{A^{(i)}_{n}=\gamma|Q^{(i)}_{n}=\pi, \mathbf{W}_{\Omega}\}.
\end{equation}
The term $\Pr\{A^{(i)}_{n}=\gamma|Q^{(i)}_{n}=\pi, \mathbf{W}_{\Omega}\}$ is also a constant, independent of the requested message. Thus, to meet the definition in \eqref{eq:DP}, we show that possible structures of each query satisfy:
\begin{equation}
\label{eq:DP_p}
  \frac{\Pr(\pi_{n,i}|i=k_1)}{\Pr(\pi_{n,i}|i=k_2)} \leq e^{\epsilon} , \forall n, \ k_1,k_2\in[1:K],
\end{equation}
where $\Pr(\pi_{n,i}|i=k_1),$ is the probability of selecting structure $\pi_{n,i}$ when the desired message is $W_{k_1}$. 
The following lemma generalizes the condition in \eqref{eq:DP_p2} to satisfy $\epsilon-$user privacy for any $K$. It states the upper bound on the path biasing that does not violate the $\epsilon-$user privacy.

%From \eqref{eq:pi_Ns}, we notice that any structure with a certain $x_k$'s can be used to recover any desired message. This is obtained by requesting that structure with other $N-1$ structures that share the same $x_k$'s of the $K-1$ unrequested messages. This is crucial to satisfy the $\epsilon-$user privacy requirements because accessing a structure does not eliminate any of message possibilities. To prove $\epsilon-$user privacy, we  express probability $Pr\{Q^{i}_{n}=\pi,A^{i}_{n}=\alpha| \mathbf{W}_{\Omega}\}$  as
%\begin{equation}
%\label{eq:20}
%\begin{split}
 %   Pr\{Q^{i}_{n}=\pi,A^{i}_{n}&=\alpha| \mathbf{W}_{\Omega}\}=Pr\{Q^{i}_{n}=\pi| \mathbf{W}_{\Omega}\} \\ &\quad \quad \times Pr\{A^{i}_{n}=\alpha|Q^{i}_{n}=\pi, \mathbf{W}_{\Omega}\}.
  %  \end{split}
%\end{equation}
%The term $Pr\{Q^{i}_{n}=\pi| \mathbf{W}_{\Omega}\}$ depends on the path selection probability. However, for any answer to a specific structure  $\pi_{n}(\bar{x})$, the term $Pr\{A^{i}_{n}=\alpha|Q^{i}_{n}=\pi, \mathbf{W}_{\Omega}\}$ should equal a constant independent of the requested message. 
%Thus to meet this definition in \eqref{eq:DP}, it is sufficient to show that possible structures to each query must satisfy:
%\begin{align}
%  &\frac{\Pr(\pi_{n}(\bar{x})|i_1)}{\Pr(\pi_{n}(\bar{x})|i_2)} <e^{\epsilon} ,\nonumber\\
 % &\quad\quad \forall n\in[1:N],  \ \bar{x}\in[0:N-1]^K, \ i_1,i_2\in[1:K].
%\end{align}
%where $\Pr(\pi_{n}(\bar{x})|i),$ is the probability of selecting structure $\pi_{n}(\bar{x})$ when the desired message is $W_i$.  
\begin{lemma}
\label{lem:p<2}
To preserve $\epsilon-$user privacy definition of the AL-PIR, the biased probability $p$ has to satisfy the following inequality 
\begin{equation}
\label{eq:lemma-2}
    p\leq \frac{e^{\epsilon}}{Ne^{\epsilon}+N^{K}-N}.
\end{equation}
\end{lemma}
\begin{proof}
 Based on the proposed scheme,
 each structure $\pi_{n,i}$ can be selected with probability $p$ or $q$, then
\begin{equation}
  \frac{\Pr(\pi_{n,i}|i=k_1)}{\Pr(\pi_{n,i}|i=k_2)} \in\bigg\{\frac{p}{p},\frac{q}{q},\frac{p}{q},\frac{q}{p}\bigg\}\leq e^{\epsilon}.
\end{equation}
As $p\geq q$, we only need to guarantee that
 \begin{equation}
     \frac{p}{q} \leq  e^{\epsilon}.
 \end{equation}
   Substituting  \eqref{eq:sum_prob_N} in the inequality, we get
    \begin{equation}
    \begin{split}
     e^{\epsilon}\ &\geq\frac{p}{q}=\frac{(N^K-N) p}{(N^K-N) q}=\frac{(N^K-N) p}{1-Np}.
    \end{split}
 \end{equation}
 By rearranging the above inequality, we get the following:
    \begin{equation}
    \label{eq:p33}
    p\leq   \frac{e^{\epsilon}}{Ne^{\epsilon}+N^{K}-N}.  
   \end{equation}
   %\textcolor{red}{do we need both sides of the inequality?}
Equation \eqref{eq:sum_prob_N} can be used to find the following equivalent condition:
\begin{equation}
\label{eq:qcalc}
    q\geq\frac{1}{Ne^{\epsilon}+N^{K}-N}.
\end{equation}
\end{proof}

{\bf Analysis of $\delta-$DB privacy:}
We show that the proposed AL-PIR scheme satisfies the DB privacy leakage constraint in \eqref{eq:MPri}. From the previous construction, we categorized the paths into two groups: (a) $N$ paths of size $(1+ \alpha_{1}(\epsilon,\delta))L$ bits; and
(b) $N^K-N$ paths of size $\frac{N}{N-1}L$ bits. Then, the expected size of the answers, $H(A_{[1:N]}^{(i)})$ can be expressed as follows:
\begin{align}
    \label{eq:A1:N}
    H(A_{[1:N]}^{(i)})&=pN(1+ \alpha_{1}(\epsilon,\delta))L+q (N^K-N)\frac{N}{N-1}L\nonumber\\
    &\overset{(a)}{=} pN \alpha_{1}(\epsilon,\delta) L+L+\frac{N^K-N}{N-1}qL\nonumber\\
    &\overset{(b)}{=} L+pN \alpha_{1}(\epsilon,\delta) L+\frac{1-pN}{N-1}L\nonumber\\
     &= L+\frac{L}{N-1}-pN\left(\frac{1}{N-1}- \alpha_{1}(\epsilon,\delta)\right)L,
\end{align}
where $(a)$ and $(b)$ follow from \eqref{eq:sum_prob_N}. We then calculate $H(A_{[1:N]}^{(i)}|\W_{[1:K]\setminus i})$ as follows:
\begin{align}
    \label{eq:A1:N|W}
        H(A_{[1:N]}^{(i)}|\W_{[1:K]\setminus i})&\overset{(a)}{=}H(W_i,A_{[1:N]}^{(i)}|\W_{[1:K]\setminus i})\nonumber\\
        &\overset{(b)}{=}H(W_i)+H(S,A_{[1:N]}^{(i)}|\W_{[1:K]\setminus i},W_i)\nonumber\\
        &=H(W_i)+H(S)+H(A_{[1:N]}^{(i)}|\W_{[1:K]\setminus i},W_i,S)\nonumber\\
        &\overset{(c)}{=}H(W_i)+H(S)=(1+ \alpha_{1}(\epsilon,\delta))L,
\end{align}
where $(a)$ follows the correctness property in \eqref{eq:corre} whereas $(b)$ and $(c)$ hold from the fact that answers are function of messages and the shared randomness.
\begin{lemma}
\label{lem:p>}
To preserve $\delta-$DB privacy for $\delta<\delta_{1}(\epsilon)$, the biased probability $p$ has to satisfy the following inequality 
\begin{equation}
\label{eq:lemma-3}
    p\geq \frac{e^{\epsilon}}{Ne^{\epsilon}+N^{K}-N}.
\end{equation}
\end{lemma}

\begin{proof}
According to \eqref{eq:A1:N} and \eqref{eq:A1:N|W}, we can express the DB privacy leakage as
\begin{align}
         \delta L&\geq I(\W_{[1:K]\setminus i};A_{[1:N]}^{(i)})\nonumber\\
         &=H(A_{[1:N]}^{(i)})-H(A_{[1:N]}^{(i)}|\W_{[1:K]\setminus i})\nonumber\\
        &=L+\frac{L}{N-1}-pN\left(\frac{1}{N-1}- \alpha_{1}(\epsilon,\delta)\right)L-(1+ \alpha_{1}(\epsilon,\delta))L\nonumber\\
        &\overset{(a)}{=}(pN-1) \alpha_{1}(\epsilon,\delta) L+\frac{1-pN}{N-1}L\nonumber\\
        &=(1-pN)(\frac{1}{N-1}- \alpha_{1}(\epsilon,\delta))L\nonumber\\
       & \overset{(b)}{=} (1-pN)  \min\left(\frac{1}{N-1},\ \frac{e^{\epsilon}+N^{K-1}-1}{N^{K-1}-1} \ \delta\right) L,\label{eq:DB-priv}
\end{align}
where $(a)$ follows from \eqref{eq:sum_prob_N} and $(b)$ follows from \eqref{eq:alpha}. For the commonly shared randomness $S$, we have one of the following two cases:
\begin{itemize}
    \item {\textbf{No shared randomness is needed} ($\mathbf{ \alpha_{1}(\epsilon,\delta)=0}$):}\\ 
    In this case, the condition in \eqref{eq:DB-priv} can be written as follows,
    \begin{align}
    \label{eq:68''}
         \delta L&\geq  \frac{1-pN}{N-1} L \overset{(a)}{\geq} \frac{N^{K-1}-1}{(N-1)(e^{\epsilon}+N^{K-1}-1)}L =  \delta_{1}(\epsilon) L,
\end{align}
where step $(a)$ follows by applying the $\epsilon-$user privacy condition obtained in Lemma~\ref{lem:p<2}. Therefore, we obtain the bound on the DB privacy leakage $\delta\geq \delta_{1}(\epsilon)$, i.e., DB privacy leakage is maximized which covers the L-PIR model previously considered in \cite{samy2019capacity}. This case requires no condition on the biased probability $p$ as the inequality in \eqref{eq:68''} is achieved for any $p$.  We highlight that this scheme obtains a better  DB privacy compared to the perfect PIR scheme proposed in \cite{sun2017capacity}, without the need to any shared amount of randomness. The latter scheme causes a leakage of  $\frac{N^{K-1}-1}{(N-1)N^{K-1}}L$ bits.
    
    \item {\textbf{Shared randomness is needed} ($\mathbf{ \alpha_{1}(\epsilon,\delta)>0}$):}\\
    For any $ \alpha_{1}(\epsilon,\delta)>0$, we always have
    \begin{align}
        \frac{1}{N-1}>\ \frac{e^{\epsilon}+N^{K-1}-1}{N^{K-1}-1} \ \delta.
    \end{align}
    From \eqref{eq:DB-priv}, we get the following relation on $p$: %\textcolor{red}{Rephrase. Not sure what you mean by according to alpha1, we can achieve delta privacy.}
\begin{align}
    (1-pN)  \leq \frac{N^{K-1}-1}{e^{\epsilon}+N^{K-1}-1},
\end{align}
which leads to the proof of Lemma~\ref{lem:p>}.

\end{itemize}

\end{proof}

Lemmas~\ref{lem:p<2} and \ref{lem:p>} lead to the following necessary condition on $p$ to simultaneously satisfy the $\epsilon-$user privacy and $\delta-$DB privacy definitions.
\begin{lemma}
\label{lem:p=}
To preserve $\epsilon-$user privacy and $\delta-$DB privacy, the biased probability $p$ has to satisfy the following condition with equality 
\begin{equation}
\label{eq:lemma-1}
    p= \frac{e^{\epsilon}}{Ne^{\epsilon}+N^{K}-N}.
\end{equation}
\end{lemma}
\begin{proof}
For $\delta<\delta_{1}(\epsilon)$, the proof follows directly by applying Lemmas~\ref{lem:p<2} and \ref{lem:p>}. For $\delta\geq \delta_{1}(\epsilon)$, the proof follows from Lemma~\ref{lem:p<2} where we pick the maximum value of the biasing probability $p$ in order to minimize the download cost, i.e., maximize the probability of picking the paths of lower cost.
\end{proof}

{\bf Analysis of download cost:} Given that all messages are requested equiprobably, the download cost can be written as, 
\begin{align}
    D(\epsilon,\delta)&=\frac{H(A_{[1:N]}^{(i)})}{ L} \overset{(a)}{=} 1+\frac{1}{N-1}-pN\left(\frac{1}{N-1}- \alpha_{1}(\epsilon,\delta)\right)\nonumber\\
     &\overset{(c)}{=} 1+\frac{1}{N-1}-pN\min\left(\frac{1}{N-1},\ \frac{e^{\epsilon}+N^{K-1}-1}{N^{K-1}-1} \ \delta\right)\nonumber\\
     &\overset{(d)}{=} 1+\frac{1}{N-1}-\frac{e^{\epsilon}}{e^{\epsilon}+N^{K-1}-1}\times\min\left(\frac{1}{N-1},\ \frac{e^{\epsilon}+N^{K-1}-1}{N^{K-1}-1} \ \delta\right),\label{eq:Download-cost}
     %&= 1+\frac{1}{N-1}-\frac{\delta e^{\epsilon}}{N^{K-1}-1} ,\\
    \end{align}
where $(a)$ follows from \eqref{eq:A1:N}, $(b)$ comes from \eqref{eq:sum_prob_N}, $(c)$ follows \eqref{eq:alpha}, and $(d)$ is due to Lemma \ref{lem:p=}. 
According to the size of the available randomness $S$, we have one of the following two cases:
\begin{itemize}
    \item {\textbf{No shared randomness is needed} ($\mathbf{ \alpha_{1}(\epsilon,\delta)=0}$):}\\ 
    This case corresponds to $\delta \geq \delta_{1}(\epsilon)$. The download cost in \eqref{eq:DB-priv} can be written as follows,
\begin{align}
\label{eq:download-cost-1}
D(\epsilon,\delta)&= 1+\frac{1}{N-1}-\frac{e^{\epsilon}}{e^{\epsilon}+N^{K-1}-1}\times\ \frac{1}{N-1} = 1+\frac{N^{K-1}-1}{(N-1)(e^{\epsilon}+N^{K-1}-1)}\nonumber\\
&= 1+\frac{N^{K-1}}{e^{\epsilon}+N^{K-1}-1}\Bigg (\frac{1}{N}+\dots+\frac{1}{N^{K-1}} \Bigg ) =d_1(\epsilon,\delta_{1}(\epsilon)).
\end{align}
%where $d_1(\epsilon,\delta_{1}(\epsilon))$ is defined in \eqref{eq:d1_e-d}.
    
    \item {\textbf{Shared randomness is needed} ($\mathbf{ \alpha_{1}(\epsilon,\delta)>0}$):}\\
     For any $ \alpha_{1}(\epsilon,\delta)>0$, we  have
    \begin{align}
        \frac{1}{N-1}>\ \frac{e^{\epsilon}+N^{K-1}-1}{N^{K-1}-1} \ \delta.
    \end{align}
    Then, the download cost in \eqref{eq:Download-cost} can be re-expressed as
\begin{align}
\label{eq:download-cost-2}
D(\epsilon,\delta)&= 1+\frac{1}{N-1}-\frac{e^{\epsilon}}{e^{\epsilon}+N^{K-1}-1}\times\ \frac{e^{\epsilon}+N^{K-1}-1}{N^{K-1}-1}\delta = 1+\frac{N^{K-1}-1}{(N-1)(e^{\epsilon}+N^{K-1}-1)}\nonumber\\
&= 1+\frac{1}{N-1}-\frac{\delta e^{\epsilon}}{N^{K-1}-1} =d_1(\epsilon,\delta).
\end{align}
\end{itemize}
Both cases in \eqref{eq:download-cost-1} and \eqref{eq:download-cost-2}  yield  the upper bound in \eqref{eq:UB} for the  download cost of AL-PIR and prove Theorem~\ref{Thm:UB}.

\section{Proof of Theorem \ref{Thm:LB} : Lower Bound on  $D^*(\epsilon,\delta)$}
\label{Pr2}

Without loss of generality, assume the requested message is $W_1$. We can bound   $D^*(\epsilon,\delta)$ as follows
\begin{align}
        D^*(\epsilon,\delta)&=\frac{\sum_{n=1}^N H(A_n^{(1)})}{L}\nonumber\\
        &\geq \frac{H(A_{[1:N]}^{(1)})}{L}\nonumber\\
        \label{eq:D*_boundA}
        &\geq \frac{H(A_{[1:N]}^{(1)}|Q_{[1:N]}^{(1)})}{L}.
\end{align}

To further bound $D^*(\epsilon,\delta)$, we first state the following two lemmas. Proofs of both lemmas can be found in the appendices. In Lemma \ref{Ent_diff}, we introduce the relation between the entropy of answers downloaded to retrieve different messages given a certain message. We emphasize that, under perfect privacy definitions, the entropy should be exactly the same regardless of the requested message,
\begin{equation}
H(A_n^{(k_1)}|W_{k_1},Q_n^{(k_1)})=H(A_n^{(k_2)}|W_{k_1},Q_n^{(k_2)}), \quad \forall k_1\neq k_2, \ n\in[1:N].  
\end{equation}
However, this does not hold under the $\epsilon-$user privacy definition.
\begin{lemma}
\label{Ent_diff}
Under the $\epsilon-$user privacy definition, for any $k_1$ and $k_2\in [1:K]$ and a non-negative constant $\epsilon$, we have the following inequality
\begin{equation}
\label{eq:DXY}
H(A_n^{(k_1)}|W_{k_1},Q_n^{(k_1)})\geq \frac{1}{e^{\epsilon}}H(A_n^{(k_2)}|W_{k_1},Q_n^{(k_2)}), \quad \forall k_1\neq k_2, \ n\in[1:N].  
\end{equation}
\end{lemma}
%\begin{proof}
%The proof can be found in Appendix \ref{app_lem5}.
%\end{proof}
Using Lemma \ref{Ent_diff}, we get the following recursion lemma.
\begin{lemma}
\label{lem:Ind}
For $k\in [2,K]$, we have  
\begin{equation}
\begin{split}
    H(A_{[1:N]}^{(k)}|W_{[1:k-1]},Q_{[1:N]}^{(k)})
    &\geq (1-o(L))L+\frac{1}{Ne^{\epsilon}}H(A_{[1:N]}^{(k+1)}|W_{[1:k]},Q_{[1:N]}^{(k+1)}).
    \end{split}
\end{equation}
\end{lemma}
%\begin{proof} The proof of the lemma is given in Appendix \ref{app_lem6}.
%\end{proof}
Using Lemmas \ref{Ent_diff}, and \ref{lem:Ind}, we bound $H(A_{[1:N]}^{(1)}|Q_{[1:N]}^{(1)})$ as follows
\begin{align}
H(A_{[1:N]}^{(1)}|Q_{[1:N]}^{(1)})&=
H(W_1,A_{[1:N]}^{(1)}| Q_{[1:N]}^{(1)})-H(W_1|A_{[1:N]}^{(1)}, Q_{[1:N]}^{(1)})\nonumber\\
&\overset{(a)}{=} H(W_1,A_{[1:N]}^{(1)}| Q_{[1:N]}^{(1)})-o(L)L\nonumber\\
&= H(W_1|Q_{[1:N]}^{(1)}) + H(A_{[1:N]}^{(1)}|W_1,Q_{[1:N]}^{(1)})-o(L)L\nonumber\\
&\overset{(b)}{=}(1-o(L))L+H(A_{[1:N]}^{(1)}|W_1, Q_{[1:N]}^{(1)})\nonumber\\ 
&\geq (1-o(L))L+H(A_{n}^{(1)}|W_1, Q_{[1:N]}^{(1)})\nonumber\\
&\overset{(c)}{=} (1-o(L))L+H(A_{n}^{(1)}|W_1, Q_{n}^{(1)})\nonumber\\
&\overset{(d)}{\geq} (1-o(L))L+\frac{1}{e^{\epsilon}}H(A_{n}^{(2)}|W_1, Q_{n}^{(2)})
\label{eq:converse},
\end{align}
where  $(a)$ is due to the correctness property in \eqref{eq:corre}, $(b)$ follows from the fact that the message content is independent of queries, $(c)$ comes from the fact that the answer $A_n^{(1)}$ is conditionally independent of the queries submitted to other DBs given the query $Q_n^{(1)}$, whereas $(d)$ comes from Lemma \ref{Ent_diff}.
%Based on the previous inequality, we have
%\begin{equation}
%\label{eq:33}
 %  H(A_{[1:N]}^{(1)})\geq L+e^{-\epsilon}H(A_n^{(2)}|W_1), \quad \forall n. 
%\end{equation}
The addition of the previous relation over all possible $n$'s gives us the following
\begin{equation}
   NH(A_{[1:N]}^{(1)}|Q_{[1:N]}^{(1)})\geq  N (1-o(L))L+\frac{1}{e^{\epsilon}} \sum_{n=1}^{N} H(A_n^{(2)}|W_1,Q_{n}^{(2)}). 
\end{equation}
%Before we bound the download cost, we propose the following lemma.
Dividing by $N$,
\begin{align}
   &H(A_{[1:N]}^{(1)}|Q_{[1:N]}^{(1)})\nonumber\\&\geq (1-o(L))L+ \frac{1}{Ne^{\epsilon}} \sum_{n=1}^{N} H(A_n^{(2)}|W_1,Q_{n}^{(2)})\nonumber\\ 
   &\geq (1-o(L))L+ \frac{1}{Ne^{\epsilon}} \sum_{n=1}^{N} H(A_n^{(2)}|W_1,Q_{[1:N]}^{(2)})\nonumber\\ 
   &\geq (1-o(L))L+ \frac{1}{Ne^{\epsilon}}  H(A_{[1:N]}^{(2)}|W_1,Q_{[1:N]}^{(2)})\nonumber\\ 
   &\overset{(a)}{=} (1-o(L))L+\frac{1}{Ne^{\epsilon}}  H(A_{[1:N]}^{(2)}|W_1,Q_{[1:N]}^{(2)})+ \frac{1}{Ne^{\epsilon}}  H(W_2|A_{[1:N]}^{(2)},W_1,Q_{[1:N]}^{(2)})-\frac{o(L)L}{Ne^{\epsilon}}\nonumber\\
   &{=} (1-o(L))L+ \frac{1}{Ne^{\epsilon}}  H(W_2,A_{[1:N]}^{(2)}|W_1,Q_{[1:N]}^{(2)})-\frac{o(L)L}{Ne^{\epsilon}}\nonumber\\
   &= (1-o(L))L+\frac{1}{Ne^{\epsilon}}  H(W_2|W_1,Q_{[1:N]}^{(2)})+\frac{1}{Ne^{\epsilon}}H(A_{[1:N]}^{(2)}|W_1,W_2,Q_{[1:N]}^{(2)})-\frac{o(L)L}{Ne^{\epsilon}}\nonumber\\
   &\overset{(b)}{=} (1-o(L))L+\frac{1}{Ne^{\epsilon}}  L+\frac{1}{Ne^{\epsilon}}H(A_{[1:N]}^{(2)}|W_1,W_2,Q_{[1:N]}^{(2)})-\frac{o(L)L}{Ne^{\epsilon}}\nonumber\\
   \label{eq:34}
   &= (1-o(L))L+\frac{1}{Ne^{\epsilon}}(1-o(L))  L+\frac{1}{Ne^{\epsilon}}H(A_{[1:N]}^{(2)}|W_1,W_2,Q_{[1:N]}^{(2)}),
\end{align}
where $(a)$ comes from the correctness property in \eqref{eq:corre}, and $(b)$ is due to the message independence. Following the same iterative process used in \cite{sun2017capacity}, and invoking the recursion property in Lemma \ref{lem:Ind}, we get
%Next, we apply the recursion relation in Lemma \ref{lem:Ind}  to get a bound on $H(A_{[1:N]}^{(1)})$.
\begin{align}
\label{eq:50}
    &H(A_{[1:N]}^{(1)}|Q_{[1:N]}^{(1)}) \nonumber\\&\geq (1+ \frac{1}{Ne^{\epsilon}}+\dots+\frac{1}{(Ne^{\epsilon})^{K-1}})(1-o(L))L+\frac{1}{(Ne^{\epsilon})^{K-1}} \ H(A_{[1:N]}^{(K)}|W_{[1:K]},Q_{[1:N]}^{(K)})\nonumber\\
    &\overset{(a)}{=}(1-o(L))L+\frac{(Ne^{\epsilon})^{K-1}-1}{(Ne^{\epsilon})^{K-1}(Ne^{\epsilon}-1)} \ (1-o(L))L+\frac{1}{(Ne^{\epsilon})^{K-1}} \ H(A_{[1:N]}^{(K)}|W_{[1:K]},Q_{[1:N]}^{(K)}),
\end{align}
where $(a)$ follows from the rule of finite sum of geometric series. Under the L-PIR model presented in \cite{samy2019capacity}, the term $H(A_{[1:N]}^{(K)}|W_{[1:K]},Q_{[1:N]}^{(K)})$ is replaced by zero as answers are functions of only the $K$ messages. However, this does not hold in the presence of common randomness. From the $\delta-$DB privacy definition in \eqref{eq:MPri}, we get the following:
    \begin{align}
        H(A_{[1:N]}^{(K)}|W_{[1:K]},Q_{[1:N]}^{(K)})&=H(A_{[1:N]}^{(K)}|W_{[1:K-1]},Q_{[1:N]}^{(K)})-I(W_K;A_{[1:N]}^{(K)}|W_{[1:K-1]},Q_{[1:N]}^{(K)})\nonumber\\
        &\overset{(a)}{=}H(A_{[1:N]}^{(K)}|W_{[1:K-1]},Q_{[1:N]}^{(K)})-H(W_K)\nonumber\\
        &=H(A_{[1:N]}^{(K)}|Q_{[1:N]}^{(K)})-I(A_{[1:N]}^{(K)};W_{[1:K-1]}|Q_{[1:N]}^{(K)})-H(W_K)\nonumber\\
        &= H(A_{[1:N]}^{(K)}|Q_{[1:N]}^{(K)})-I(A_{[1:N]}^{(K)},Q_{[1:N]}^{(K)};W_{[1:K-1]})-H(W_K)\nonumber\\
        &\geq H(A_{[1:N]}^{(K)}|Q_{[1:N]}^{(K)})-\delta L-L,
    \end{align}
    where $(a)$ follows since all messages are independent and $W_K$ is a deterministic function of $A_{[1:N]}^{(K)}$.
By symmetry, we can assume that $
    H(A_{[1:N]}^{(1)}|Q_{[1:N]}^{(1)})=H(A_{[1:N]}^{(K)}|Q_{[1:N]}^{(K)}).
$
Then,
\begin{equation}
    \begin{split}
        H(A_{[1:N]}^{(K)}|W_{[1:K]},Q_{[1:N]}^{(K)})
        &\geq H(A_{[1:N]}^{(1)}|Q_{[1:N]}^{(1)})-\delta L-L.
    \end{split}
\end{equation}
Since $H(A_{[1:N]}^{(K)}|W_{[1:K]},Q_{[1:N]}^{(K)}) \geq 0 $, we obtain
\begin{equation}
\label{eq:delta-lb}
    \begin{split}
        H(A_{[1:N]}^{(K)}|W_{[1:K]},Q_{[1:N]}^{(K)})
        &\geq \max\left(0,\ H(A_{[1:N]}^{(1)}|Q_{[1:N]}^{(1)})-\delta L-L\right).
    \end{split}
\end{equation}
Next, we can express \eqref{eq:50} using \eqref{eq:delta-lb} as
 \begin{align}
    &H(A_{[1:N]}^{(1)}|Q_{[1:N]}^{(1)}) \geq \nonumber\\ \label{eq:53}
    &  \ (1+\frac{(Ne^{\epsilon})^{K-1}-1}{(Ne^{\epsilon})^{K-1}(Ne^{\epsilon}-1)})(1-o(L))L+\frac{1}{(Ne^{\epsilon})^{K-1}}\max\left(0, H(A_{[1:N]}^{(1)}|Q_{[1:N]}^{(1)})-\delta L-L\right).
    \end{align}
Dividing by $L$ and allowing it to approach $\infty$, we get
 \begin{align}
   & \frac{H(A_{[1:N]}^{(1)}|Q_{[1:N]}^{(1)})}{L} \geq
 1+\frac{(Ne^{\epsilon})^{K-1}-1}{(Ne^{\epsilon})^{K-1}(Ne^{\epsilon}-1)}+\frac{1}{(Ne^{\epsilon})^{K-1}}\max\left(0, \frac{H(A_{[1:N]}^{(1)}|Q_{[1:N]}^{(1)})}{L}-\delta -1\right).
   \label{eq:91''}
%   &=\max\left(1+\frac{(Ne^{\epsilon})^{K-1}-1}{(Ne^{\epsilon})^{K-1}(Ne^{\epsilon}-1)}, 1+\frac{(Ne^{\epsilon})^{K-1}-1}{(Ne^{\epsilon})^{K-1}(Ne^{\epsilon}-1)}+\frac{1}{(Ne^{\epsilon})^{K-1}} (\frac{H(A_{[1:N]}^{(1)}|Q_{[1:N]}^{(1)})}{L}-\delta -1)\right).
    \end{align}
   Following \eqref{eq:91''}, the following two inequalities are true
         \begin{align}
    \frac{H(A_{[1:N]}^{(1)}|Q_{[1:N]}^{(1)})}{L}&\geq 1+\frac{(Ne^{\epsilon})^{K-1}-1}{(Ne^{\epsilon})^{K-1}(Ne^{\epsilon}-1)}\quad\quad\quad\quad\quad\quad\quad\quad\text{and,} \label{eq:92''}\\
   \label{eq:93''}
   \frac{H(A_{[1:N]}^{(1)}|Q_{[1:N]}^{(1)})}{L}&\geq 1+\frac{(Ne^{\epsilon})^{K-1}-1}{(Ne^{\epsilon})^{K-1}(Ne^{\epsilon}-1)}+\frac{1}{(Ne^{\epsilon})^{K-1}} (\frac{H(A_{[1:N]}^{(1)}|Q_{[1:N]}^{(1)})}{L}-\delta -1).
    \end{align}
The inequality in  \eqref{eq:93''} can be rearranged as

\begin{align}
    \frac{(Ne^{\epsilon})^{K-1}-1}{(Ne^{\epsilon})^{K-1}} \ \frac{H(A_{[1:N]}^{(1)}|Q_{[1:N]}^{(1)})}{L} &\geq\ \ 1+\frac{(Ne^{\epsilon})^{K-1}-1}{(Ne^{\epsilon})^{K-1}(Ne^{\epsilon}-1)} \ -\frac{1}{(Ne^{\epsilon})^{K-1}} \ ( \delta +1)\nonumber\\
    &=\frac{(Ne^{\epsilon})^{K-1}-1}{(Ne^{\epsilon})^{K-1}} +\frac{(Ne^{\epsilon})^{K-1}-1}{(Ne^{\epsilon})^{K-1}(Ne^{\epsilon}-1)} \ -\frac{\delta}{(Ne^{\epsilon})^{K-1}},\nonumber\\
    \frac{H(A_{[1:N]}^{(1)}|Q_{[1:N]}^{(1)})}{L} &\geq \ 1+\frac{1}{Ne^{\epsilon}-1} -\frac{\delta}{(Ne^{\epsilon})^{K-1}-1}.\label{eq:88'}
    \end{align}
%Rearranging the inequality, we get
 %\begin{equation}
 %\label{eq:88'}
%\begin{split}
 %  \frac{H(A_{[1:N]}^{(1)}|Q_{[1:N]}^{(1)})}{L} \geq \ 1+\frac{1}{Ne^{\epsilon}-1} -\frac{\delta}{(Ne^{\epsilon})^{K-1}-1}.
  %  \end{split}
%\end{equation}
From \eqref{eq:92''} and \eqref{eq:88'}, we get

 \begin{equation}
 \label{eq:96'}
\begin{split}
   \frac{H(A_{[1:N]}^{(1)}|Q_{[1:N]}^{(1)})}{L} \geq \max\left(1+\frac{(Ne^{\epsilon})^{K-1}-1}{(Ne^{\epsilon})^{K-1}(Ne^{\epsilon}-1)}, 1+\frac{1}{Ne^{\epsilon}-1} -\frac{\delta}{(Ne^{\epsilon})^{K-1}-1}\right).
    \end{split}
\end{equation}
Substituting by \eqref{eq:96'} in \eqref{eq:D*_boundA}, we can lower bound $D^*(\epsilon,\delta)$ as
\begin{align}
\label{eq:89'}
     D^*(\epsilon,\delta)&\geq \frac{H(A_{[1:N]}^{(1)}|Q_{[1:N]}^{(1)})}{L} \geq D^{\text{LB}}(\epsilon,\delta)\nonumber\\
        &= \max\left(1+\frac{(Ne^{\epsilon})^{K-1}-1}{(Ne^{\epsilon})^{K-1}(Ne^{\epsilon}-1)},\  1+\frac{1}{Ne^{\epsilon}-1} -\frac{\delta}{(Ne^{\epsilon})^{K-1}-1}\right).
    \end{align}

For a fixed $\epsilon$, $D^{\text{LB}}(\epsilon,\delta)$ is monotonically decreasing in $\delta$ until we reach $\delta=\delta_2(\epsilon)=\frac{(Ne^{\epsilon})^{K-1}-1}{(Ne^{\epsilon}-1)(Ne^{\epsilon})^{K-1}}$ at which
\begin{equation}
1+\frac{(Ne^{\epsilon})^{K-1}-1}{(Ne^{\epsilon})^{K-1}(Ne^{\epsilon}-1)}= 1+\frac{1}{Ne^{\epsilon}-1} -\frac{\delta}{(Ne^{\epsilon})^{K-1}-1}.    
\end{equation}
After this point, $D^{\text{LB}}(\epsilon,\delta)$ is fixed at the value $1+\frac{(Ne^{\epsilon})^{K-1}-1}{(Ne^{\epsilon})^{K-1}(Ne^{\epsilon}-1)}$. Then, we can alternatively represent $D^{\text{LB}}(\epsilon,\delta)$ as 
\begin{equation}
    D^{*}(\epsilon,\delta)\geq D^{\text{LB}}(\epsilon,\delta)= \begin{cases} 1+\frac{1}{Ne^{\epsilon}-1} -\frac{\delta}{(Ne^{\epsilon})^{K-1}-1}= d_2(\epsilon,\delta),&  0\leq \delta < \delta_2(\epsilon),\\
    1+\frac{(Ne^{\epsilon})^{K-1}-1}{(Ne^{\epsilon})^{K-1}(Ne^{\epsilon}-1)}= d_2(\epsilon,\delta_2(\epsilon)), & \delta\geq \delta_2(\epsilon).
    \end{cases}
\end{equation}
This proves the lower bound on $D^{*}(\epsilon,\delta)$  in Theorem~\ref{Thm:LB}.

\subsection{Required amount of shared randomness}
In this section, we prove the lower bound in Theorem~\ref{Thm:LB} on the required amount of shared randomness to achieve the minimum download cost derived in \eqref{eq:89'}. From the $\delta-$DB privacy in \eqref{eq:MPri}, given a requested message $W_k$, we get
\begin{align}
\label{eq:90'}
    \delta L&\geq I(A_{[1:N]}^{(k)},Q_{[1:N]}^{(k)};W_{[1:K]\setminus k})\nonumber\\ 
    &= I(A_{[1:N]}^{(k)};W_{[1:K]\setminus k}|Q_{[1:N]}^{(k)})\nonumber\\
    &= H(A_{[1:N]}^{(k)}|Q_{[1:N]}^{(k)})- H(A_{[1:N]}^{(k)}|W_{[1:K]\setminus k},Q_{[1:N]}^{(k)})\nonumber\\
    &= H(A_{[1:N]}^{(k)}|Q_{[1:N]}^{(k)})- H(W_k,A_{[1:N]}^{(k)}|W_{[1:K]\setminus k},Q_{[1:N]}^{(k)})+H(W_k|A_{[1:N]}^{(k)},W_{[1:K]\setminus k},Q_{[1:N]}^{(k)})\nonumber\\
    &\overset{(a)}{=}H(A_{[1:N]}^{(k)}|Q_{[1:N]}^{(k)})- H(W_k,A_{[1:N]}^{(k)}|W_{[1:K]\setminus k},Q_{[1:N]}^{(k)})+o(L)L\nonumber\\
    &=H(A_{[1:N]}^{(k)}|Q_{[1:N]}^{(k)})- H(W_k|W_{[1:K]\setminus k},Q_{[1:N]}^{(k)})-H(A_{[1:N]}^{(k)}|W_{[1:K]},Q_{[1:N]}^{(k)})+o(L)L\nonumber\\
    &=H(A_{[1:N]}^{(k)}|Q_{[1:N]}^{(k)})- L-H(A_{[1:N]}^{(k)}|W_{[1:K]},Q_{[1:N]}^{(k)})+o(L)L\nonumber\\
    &=H(A_{[1:N]}^{(k)}|Q_{[1:N]}^{(k)})- (1-o(L))L-H(A_{[1:N]}^{(k)}|W_{[1:K]},S,Q_{[1:N]}^{(k)})-I(S;A_{[1:N]}^{(k)}|W_{[1:K]},Q_{[1:N]}^{(k)})\nonumber\\
    &\overset{(b)}{=}H(A_{[1:N]}^{(k)}|Q_{[1:N]}^{(k)})- (1-o(L))L-I(S;A_{[1:N]}^{(k)}|W_{[1:K]},Q_{[1:N]}^{(k)})\nonumber\\
    &=H(A_{[1:N]}^{(k)}|Q_{[1:N]}^{(k)})- (1-o(L))L-H(S|W_{[1:K]},Q_{[1:N]}^{(k)})+H(S|A_{[1:N]}^{(k)},W_{[1:K]},Q_{[1:N]}^{(k)})\nonumber\\
     &\overset{(c)}{=}H(A_{[1:N]}^{(k)}|Q_{[1:N]}^{(k)})- (1-o(L))L-H(S)+H(S|A_{[1:N]}^{(k)},W_{[1:K]},Q_{[1:N]}^{(k)})\nonumber\\
    &\geq H(A_{[1:N]}^{(k)}|Q_{[1:N]}^{(k)})- (1-o(L))L-H(S),
\end{align}
where $(a)$ follows from the correctness property, $(b)$ comes from the fact that answers are function of the $K$ messages and the common randomness  $S$, and $(c)$ is because the common randomness $S$ is independent of the $K$ messages. Dividing by $L$, allowing it to approach $\infty$, and substituting by \eqref{eq:88'} in \eqref{eq:90'}, we get
\begin{align}
\label{eq:91'}
    \delta  &\geq  1+\frac{1}{Ne^{\epsilon}-1} -\frac{\delta}{(Ne^{\epsilon})^{K-1}-1}- 1-\frac{H(S)}{L}\nonumber\\
    &=\frac{1}{Ne^{\epsilon}-1} \ -\frac{\delta}{(Ne^{\epsilon})^{K-1}-1}-\frac{H(S)}{L}.
\end{align}
Rearranging the inequality, we get the following bound on $\alpha(\epsilon,\delta)$,
\begin{align}
\label{eq:92'}
    \frac{H(S)}{L}=\alpha(\epsilon,\delta) \geq \frac{1}{Ne^{\epsilon}-1} \ -\frac{\delta}{(Ne^{\epsilon})^{K-1}-1} -\delta,
\end{align}
which is also a valid bound on the optimal common randomness $\alpha^*(\epsilon,\delta)$.
Then, following that $\alpha^*(\epsilon,\delta)\geq 0$, we obtain the following bound,
\begin{align}
    \alpha^*(\epsilon,\delta) \geq \alpha_2(\epsilon,\delta) &=\max\left(0,\  \frac{1}{Ne^{\epsilon}-1} \ -\frac{(Ne^{\epsilon})^{K-1}}{(Ne^{\epsilon})^{K-1}-1} \delta\right)\nonumber\\
           &=\begin{cases}
     \frac{1}{Ne^{\epsilon}-1} \ -\frac{(Ne^{\epsilon})^{K-1}}{(Ne^{\epsilon})^{K-1}-1} \ \delta, & 0\leq \delta < \delta_{2}(\epsilon),\\
      0, &  \delta > \delta_{2}(\epsilon),
      \end{cases}\label{eq:93'}
\end{align}
which completes the proof of the lower bound on the optimal common randomness size in Theorem~\ref{Thm:LB}.
%\textcolor{red}{Note (Islam): Here, the amount of common randomness is different than that used in the scheme. However, if we plug the upper bound instead of $H(A_{[1:N]}^{(k)}|Q_{[1:N]}^{(k)})$ in \eqref{eq:90'}, we get the same amount. Should we state that?}

\vspace{-10pt}
\section{Conclusions}
\vspace{-5pt}
\label{Conc}
We studied the AL-PIR problem that relaxes the perfect privacy requirements for both user and DB privacy. The allowed leakage is asymmetric allowing for different privacy leakage in each direction. We showed that allowing privacy leakage provides an opportunity to improve the optimal download cost. We introduced an AL-PIR scheme that gives an upper bound on the optimal download cost for arbitrary leakage budgets. We investigated possible tradeoffs that stem by adjusting the level of privacy at both user and DB sides. We further obtained a lower bound on the download cost and showed that the multiplicative gap between the upper and lower bounds is bounded by $\frac{N-e^{-\epsilon}}{N-1}$, i.e., our AL-PIR scheme is optimal for perfect user privacy, $\epsilon =0$, and is optimal within a gap of at most $\frac{N}{N-1}$ for any $\epsilon$.

%\textcolor{red}{Go over all references and fix them. Present all conference papers as in Proceeddings of the IEEE International Symposium on... Or the in Proceedings of the GLOBECOM Conference. Remove the year from within the conference name, and the trailing IEEE after the conference.}

\bibliographystyle
{IEEEtran}
{\small  \bibliography{IEEEabrv,paper_arxiv}}

\begin{appendices}

\section{Proof of Corollary \ref{Thm:Gap}}
\label{app_gap}

We first notice that for any $\epsilon>0$, we have $\delta_1(\epsilon)\geq \delta_2(\epsilon)$. This follows as we can express  $d_1(\epsilon,\delta_1(\epsilon))=1+\delta_1(\epsilon)$ and $d_2(\epsilon,\delta_2(\epsilon))=1+\delta_2(\epsilon)$. Then, from Theorems \ref{Thm:UB} and \ref{Thm:LB} and for any $\delta\geq\max\left(\delta_1(\epsilon),\delta_2(\epsilon)\right)$, $D^*(\epsilon,\delta)$ can be bounded as follows
 \begin{equation}
 \label{eq:delta1>2}
   1+\delta_2(\epsilon)=d_2(\epsilon,\delta_2(\epsilon))=D^{LB}(\epsilon,\delta)\leq D^*(\epsilon,\delta)\leq D^{UB}(\epsilon,\delta)=  d_1(\epsilon,\delta_1(\epsilon))=1+\delta_1(\epsilon),
 \end{equation}
which proves that $\delta_1(\epsilon)$ must be greater than or equal $\delta_2(\epsilon)$ for any value of $\epsilon\geq 0$.

Following that, we can write the multiplicative gap ratio between the upper and lower bounds on $D^*(\epsilon,\delta)$ given in \eqref{eq:UB} and \eqref{eq:LB} as follows:
\begin{align}
\label{eq:proof-gap}
\frac{D^{\text{UB}}(\epsilon,\delta)}{D^{\text{LB}}(\epsilon,\delta)} &%= \frac{1+\frac{1}{N-1}-\frac{\delta e^{\epsilon}}{N^{K-1}-1}}{1+\frac{1}{Ne^{\epsilon}-1}-\frac{\delta }{(Ne^{\epsilon})^{K-1}-1}}
%\overset{(a)}{=}
=\begin{cases} \frac{\gamma_1(\epsilon)-\gamma_2(\epsilon)\delta}{\gamma_3(\epsilon)-\gamma_4(\epsilon)\delta}, \quad\quad \delta<\delta_2(\epsilon), \\ \\
\frac{\gamma_1(\epsilon)-\gamma_2(\epsilon)\delta}{\gamma_3(\epsilon)-\gamma_4(\epsilon)\delta_2(\epsilon)}, \quad \delta_2(\epsilon)\leq\delta<\delta_1(\epsilon), \\ \\
\frac{\gamma_1(\epsilon)-\gamma_2(\epsilon)\delta_1(\epsilon)}{\gamma_3(\epsilon)-\gamma_4(\epsilon)\delta_2(\epsilon)},\quad \delta\geq\delta_1(\epsilon),\\
\end{cases}
%\nonumber\\
%&\overset{(b)}{\leq}\frac{c_1}{c_3} =  \frac{1+\frac{1}{N-1}}{1+\frac{1}{Ne^{\epsilon}-1}}
%=\frac{N-e^{-\epsilon}}{N-1},
\end{align}
where  we have $\gamma_1(\epsilon) = 1+\frac{1}{N-1}$, $\gamma_2(\epsilon) = \frac{e^{\epsilon}}{N^{K-1}-1}$, $\gamma_3(\epsilon) = 1+\frac{1}{Ne^{\epsilon}-1}$, and $\gamma_4(\epsilon)= \frac{1 }{(Ne^{\epsilon})^{K-1}-1}$. 
Then, we can upper bound \eqref{eq:proof-gap} as follows,
\begin{align}
\label{eq:proof-gap-ub1}
\frac{D^{\text{UB}}(\epsilon,\delta)}{D^{\text{LB}}(\epsilon,\delta)} &
\leq \begin{cases} \frac{\gamma_1(\epsilon)-\gamma_2(\epsilon)\delta}{\gamma_3(\epsilon)-\gamma_4(\epsilon)\delta}, \quad\quad \delta<\delta_2(\epsilon), \\ \\
\frac{\gamma_1(\epsilon)-\gamma_2(\epsilon)\delta_2(\epsilon)}{\gamma_3(\epsilon)-\gamma_4(\epsilon)\delta_2(\epsilon)}, \quad \delta\geq \delta_2(\epsilon).
\end{cases}
\end{align}

For any $\delta$, we have the bound $\frac{\gamma_1(\epsilon)-\gamma_2(\epsilon)\delta}{\gamma_3(\epsilon)-\gamma_4(\epsilon)\delta}\leq \frac{\gamma_1(\epsilon)}{\gamma_3(\epsilon)}$ valid  when $\frac{\gamma_1(\epsilon)\gamma_4(\epsilon)}{\gamma_2(\epsilon)\gamma_3(\epsilon)}\leq 1$.
We can prove that $\frac{\gamma_1(\epsilon)\gamma_4(\epsilon)}{\gamma_2(\epsilon)\gamma_3(\epsilon)}\leq 1$ in the following:
\begin{align}
\label{eq:c1c4/c2c3}
\frac{\gamma_1(\epsilon)\gamma_4(\epsilon)}{\gamma_2(\epsilon)\gamma_3(\epsilon)} = \frac{(1+\frac{1}{N-1})\frac{1}{(Ne^{\epsilon})^{K-1}-1}}{(1+\frac{1}{Ne^{\epsilon}-1})\frac{e^{\epsilon}}{N^{K-1}-1}} 
= \frac{\frac{N^{K-1}-1}{N-1}}{\frac{(Ne^{\epsilon})^{K-1}-1}{Ne^{\epsilon}-1}}.{e^{-2\epsilon}} = \frac{\sum_{k=0}^{K-2}N^k}{\sum_{k=0}^{K-2}(Ne^{\epsilon})^k}.{e^{-2\epsilon}}\leq {e^{-2\epsilon}} \leq 1,
\end{align}
for any $\epsilon\geq 0$ with equality when $\epsilon = 0$.
Eventually, we can bound the multiplicative gap ratio for any value of $\delta$ as 
\begin{align}
\frac{D^{\text{UB}}(\epsilon,\delta)}{D^{\text{LB}}(\epsilon,\delta)} \leq\frac{\gamma_1(\epsilon)}{\gamma_3(\epsilon)} =  \frac{1+\frac{1}{N-1}}{1+\frac{1}{Ne^{\epsilon}-1}}
=\frac{N-e^{-\epsilon}}{N-1}.
\end{align}

\section{Proof of Lemma \ref{Ent_diff}}
\label{app_lem5}
Assume that $A_n^{(k_1)}$ the answer of any DB$_n$, given any requested message $k_1\in[1:K]$, can take one of  $T$ different structures. Each of them is requested by a certain query, i.e., $Q_n^{(k_1)}$ also takes $T$ different forms. Let $\pi_t$ and $\gamma(\pi_t)$ be the $t^{th}$ form that $Q_n^{(k_1)}$ and $A_n^{(k_1)}$ can take, respectively.   Then, $H(A_n^{(k_1)}|W_{k_1},Q_n^{(k_1)})$ can be written as 
\begin{align}
    H(A_n^{(k_1)}|W_{k_1},Q_n^{(k_1)})&=\sum_{t=1}^{T}\Pr(Q_n^{(k_1)}=\pi_t)H(A_n^{(k_1)}|W_{k_1},Q_n^{(k_1)}=\pi_t)\nonumber\\
    &=\sum_{t=1}^{T}\Pr(Q_n^{(k_1)}=\pi_t)H(A_n^{(k_1)}=\gamma(\pi_t)|W_{k_1})\nonumber\\
        &\overset{(a)}{=}\sum_{t=1}^{T}\Pr(Q_n^{(k_1)}=\pi_t)H(A_n^{(k_2)}=\gamma(\pi_t)|W_{k_1})\nonumber\\
        &\overset{(b)}{\geq}\sum_{t=1}^{T}e^{-\epsilon}\Pr(Q_n^{(k_2)}=\pi_t)H(A_n^{(k_2)}=\gamma(\pi_t)|W_{k_1})\nonumber\\
        &=\frac{1}{e^{\epsilon}}\sum_{t=1}^{T}\Pr(Q_n^{(k_2)}=\pi_t)H(A_n^{(k_2)}=\gamma(\pi_t)|W_{k_1})\nonumber\\
         &=\frac{1}{e^{\epsilon}}\sum_{t=1}^{T}\Pr(Q_n^{(k_2)}=\pi_t)H(A_n^{(k_2)}=\gamma(\pi_t)|W_{k_1},Q_n^{(k_2)}=\pi_t)\nonumber\\
        &=\frac{1}{e^{\epsilon}}H(A_n^{(k_2)}|W_{k_1},Q_n^{(k_2)}),
    \end{align}
where $(a)$ follows from the fact that the entropy of certain answer structure $\gamma(\pi_t)$ is independent of the requested message, it only depends on the query form $\pi_t$. Whereas, $(b)$ comes from the definition in \eqref{eq:DP} and the corresponding interpretation in \eqref{eq:DP_p}.

\vspace{-10pt}
\section{Proof of Lemma \ref{lem:Ind}}
\label{app_lem6} \vspace{-10pt}
We can bound $H(A_{[1:N]}^{(k)}|W_{[1:k-1]},Q_{[1:N]}^{(k)})$ as follows:
\begin{align}
    &H(A_{[1:N]}^{(k)}|W_{[1:k-1]},Q_{[1:N]}^{(k)})\nonumber
    \\&\overset{(a)}{=}H(A_{[1:N]}^{(k)}|W_{[1:k-1]},Q_{[1:N]}^{(k)})+H(W_k|A_{[1:N]}^{(k)},W_{[1:k-1]},Q_{[1:N]}^{(k)})-o(L)L \nonumber\\
    &=H(W_k,A_{[1:N]}^{(k)}|W_{[1:k-1]},Q_{[1:N]}^{(k)})-o(L)L  \nonumber\\
    &=H(W_k|W_{[1:k-1]},Q_{[1:N]}^{(k)})+H(A_{[1:N]}^{(k)}|W_{[1:k]},Q_{[1:N]}^{(k)})-o(L)L \nonumber\\
    &  =L+H(A_{[1:N]}^{(k)}|W_{[1:k]},Q_{[1:N]}^{(k)})-o(L)L \nonumber\\
    &  =(1-o(L))L+H(A_{[1:N]}^{(k)}|W_{[1:k]},Q_{[1:N]}^{(k)}) \nonumber\\
    &\geq (1-o(L))L+H(A_n^{(k)}|W_{[1:k]},Q_{[1:N]}^{(k)})\nonumber\\
    \label{eq:97}
    &\overset{(b)}{=} (1-o(L))L+H(A_n^{(k)}|W_{[1:k]},Q_{n}^{(k)}), \quad \forall n\in[1:N],
    \end{align}
where $(a)$ is due to the correctness property in \eqref{eq:corre}, and (b) comes from the fact that the answer $A_n^{(k)}$ is conditionally independent of the queries submitted to other DBs given the query $Q_n^{(k)}$. By adding the relation in \eqref{eq:97} over all possible $n$'s and dividing by $N$, we get the following:
\begin{align}
    H(A_{[1:N]}^{(k)}|W_{[1:k-1]},Q_{[1:N]}^{(k)})&
    \geq (1-o(L))L+\frac{1}{N}\sum _{n=1}^{N} H(A_n^{(k)}|W_{[1:k]},Q_{n}^{(k)})\nonumber\\
        &\overset{(a)}{\geq} \  (1-o(L))L+\frac{1}{Ne^{\epsilon}}\sum _{n=1}^{N} H(A_n^{(k+1)}|W_{[1:k]},Q_{n}^{(k+1)})\nonumber\\
         &= \  (1-o(L))L+\frac{1}{Ne^{\epsilon}}\sum _{n=1}^{N} H(A_n^{(k+1)}|W_{[1:k]},Q_{[1:N]}^{(k+1)})\nonumber\\
        &\ {\geq} \  (1-o(L))L+\frac{1}{Ne^{\epsilon}} H(A_{[1:N]}^{(k+1)}|W_{[1:k]},Q_{[1:N]}^{(k+1)}),
    \end{align}
    where (a) follows using similar steps as in the proof of Lemma \ref{Ent_diff}.
    
    \vspace{-10pt}
    \section{Proof of Proposition \ref{prob:N=1}}
\label{pr3}
\vspace{-10pt}
Here, we prove the bound in proposition \ref{prob:N=1} for $N=1$. We show that the relaxed privacy conditions have no benefits when there is only one database even if we ignore the DB privacy leakage constraint ($\delta= K-1$). Assuming that the requested message is $W_1$,  we lower bound $D_{\epsilon,\delta}$ as follows:
\begin{align}
       D_{\epsilon,\delta}= H(A_1^{(1)})&\geq H(A_1^{(1)}|Q_1^{(1)})\nonumber\\
    &= H(W_1,A_1^{(1)}|Q_1^{(1)}) -H(W_1|A_1^{(1)},Q_1^{(1)})\nonumber\\
    &\overset{(a)}{=}H(W_1|Q_1^{(1)})+H(A_1^{(1)}|W_1,Q_1^{(1)})-o(L)L\nonumber\\
    \label{eq:converse,N=1} &=(1-o(L))L+H(A_1^{(1)}|W_1,Q_1^{(1)}),
\end{align}
where $(a)$ follows the correctness property in \eqref{eq:corre}.
%Next, we show that 
%\begin{equation}
 %  H(A_1^{(1)}|W_1)=H(A_1^{(j)}|W_1), \quad \forall i,j\in [2:K], 
%\end{equation}
%even under the relaxed privacy conditions. 
Let there be $T$ different structures, $\pi_1,\dots,\pi_T$, the query sent to the databases can take. For each structure $\pi_t$, $t\in[1:T]$, the answer is on the form of $\gamma(\pi_t)$ then we get, for $j\in[2:K]$,
\begin{align}
  &H(A_1^{(1)}|W_1,Q_1^{(1)})-H(A_1^{(j)}|W_1,Q_1^{(j)})\nonumber\\&=\Big(H(A_1^{(1)},W_1|Q_1^{(1)})-H(W_1|Q_1^{(1)})\Big)-\Big(H(A_1^{(j)},W_1|Q_1^{(j)})-H(W_1|Q_1^{(j)})\Big)\nonumber\\
  &\Big(H(A_1^{(1)},W_1|Q_1^{(1)})-H(W_1)\Big)-\Big(H(A_1^{(j)},W_1|Q_1^{(j)})-H(W_1)\Big)\nonumber\\
  &=H(A_1^{(1)},W_1|Q_1^{(1)})-H(A_1^{(j)},W_1|Q_1^{(j)})\nonumber\\
  &=H(A_1^{(1)}|Q_1^{(1)})+H(W_1|A_1^{(1)},Q_1^{(1)})-H(A_1^{(j)}|Q_1^{(j)})-H(W_1|A_1^{(j)},Q_1^{(j)})\nonumber\\
  &\overset{(a)}{=}H(A_1^{(1)}|Q_1^{(1)})+o(L)-H(A_1^{(j)}|Q_1^{(j)})-\sum_{t=1}^{T}\Pr(Q_1^{(j)}=\pi_t)H(W_1|A_1^{(j)},Q_1^{(j)}=\pi_t)\nonumber\\
   \label{eq:47}
   &=H(A_1^{(1)}|Q_1^{(1)})+o(L)-H(A_1^{(j)}|Q_1^{(j)})-\sum_{t=1}^{T}\Pr(Q_1^{(j)}=\pi_t)H(W_1|A_1^{(j)}=\gamma(\pi_t)),
\end{align}
where $(a)$ also comes from  \eqref{eq:corre}.
We emphasize from the user privacy constraint in \eqref{eq:DP} that all queries or structures must be requested with non-zero probability, otherwise the constraint in \eqref{eq:DP} can not be met. This dictates that $\gamma(\pi_t)$, the answer of any structure $\pi_t$, has to fulfill the decodability conditions, i.e.,
\begin{equation}
  H(W_1|A_1^{(1)}=\gamma(\pi_t))=o(L).  
\end{equation}
As the form of the answer $\gamma(\pi_t)$ is the same regardless of the requested message, this implies that
\begin{equation}
\label{eq:49}
  H(W_1|A_1^{(j)}=\gamma(\pi_t))=H(W_1|A_1^{(1)}=\gamma(\pi_t))=o(L).  
\end{equation}
From \eqref{eq:47} and \eqref{eq:49}, we get the following
\begin{equation}
 H(A_1^{(1)}|W_1,Q_1^{(1)})=H(A_1^{(j)}|W_1,Q_1^{(j)})+H(A_1^{(1)}|Q_1^{(1)})-H(A_1^{(j)}|Q_1^{(j)}).   
\end{equation}
Assuming the symmetry across all messages, we have 
\begin{equation}
    H(A_1^{(1)}|Q_1^{(1)})=H(A_1^{(j)}|Q_1^{(j)}), \quad
\forall j\in[2:K].
\end{equation}
Using this fact, we have
\begin{equation}
\label{eq:51}
    H(A_1^{(1)}|W_1,Q_1^{(1)})=H(A_1^{(j)}|W_1,Q_1^{(j)}), \quad
\forall j\in[2:K].
\end{equation}
This allows us to write $D_{\epsilon,\delta}$ as follows 
\begin{align}
       D_{\epsilon,\delta}&\geq (1-o(L))L+H(A_1^{(2)}|W_1,Q_1^{(2)})\nonumber\\
       &=(1-o(L))L+H(W_2,A_1^{(2)}|W_1,Q_1^{(2)})-H(W_2|A_1^{(2)},W_1,Q_1^{(2)})\nonumber\\
       &=(1-o(L))L+H(W_2,A_1^{(2)}|W_1,Q_1^{(2)})-o(L)L\nonumber\\
       &=(1-2o(L))L+H(W_2|Q_1^{(2)})+H(A_1^{(2)}|W_1,W_2,Q_1^{(2)})\nonumber\\
       &=2(1-o(L))L+H(A_1^{(2)}|W_1,W_2,Q_1^{(2)}).
    \end{align}
Completing the proof inductively using  equations \eqref{eq:47} to \eqref{eq:51}, we get
\begin{align}
       D_{\epsilon,\delta}&\geq(K-1)(1-o(L))L+H(A_1^{(K)}|W_{[1:K-1]},Q_1^{(K)})\nonumber\\
       &=(K-1)(1-o(L))L+H(W_K,A_1^{(K)}|W_{[1:K-1]},Q_1^{(K)})-H(W_K|A_1^{(K)},W_{[1:K-1]},Q_1^{(K)})\nonumber\\
       &=(K-1)(1-o(L))L+H(W_K,A_1^{(K)}|W_{[1:K-1]},Q_1^{(K)})-o(L)\nonumber\\
       &\overset{(a)}{=}(K-1)(1-o(L))L+H(W_K|W_{[1:K-1]},Q_1^{(K)})+H(A_1^{(K)}|W_{[1:K-1]},Q_1^{(K)})-o(L)\nonumber\\
       &\geq K(1-o(L))L,
    \end{align}
where $(a)$ comes from the fact that the answer must be a function of the $K$ messages.
%Now, assume that the requested message is $W_i$, than we can repeat the same inductive process and get
%\begin{equation}
 %   \begin{split}
  %     D^*(\epsilon)&=KL+H(A_1^{(i)})-H(A_1^{(i+K-1 \ \text{mod} \ K)}),\\
   % \end{split}
%\end{equation}
%where ($\text{mod}$) represents the module operation. 
Dividing by $L$ and taking the limit $L\rightarrow \infty$,  we arrive at the desired lower bound:
\begin{equation}
    \begin{split}
       D^*(\epsilon,\delta)\geq K.\\
    \end{split}
\end{equation}

\end{appendices}
\end{document}